\documentclass[12pt,letterpaper]{preprint}

\usepackage{graphicx}
\usepackage{comment}

\title{Stable minimum principles for scattering states}
\author{Scott Lawrence and Yukari Yamauchi}
\report{LA-UR-26-23902}

\begin{document}
\begin{titleblock}
	Scott Lawrence\footnote{Electronic address: \texttt{srlawrence@lanl.gov}} and Yukari Yamauchi\footnote{Electronic address: \texttt{yyamauchi@lanl.gov}}\\[8pt]
\normalsize\textit{Theoretical Division, Los Alamos National Laboratory\\Los Alamos, NM 87545, USA}
\end{titleblock}
\begin{abstract}
	Quantum-mechanical scattering states are energy eigenstates obeying particular boundary conditions, whose behavior at infinity encodes the S-matrix which defines the outcoming of scattering experiments. With an eye toward numerical algorithms for computing nonrelativistic S-matrices, we present a family of \emph{stable} minimum principles for scattering states. States that approximately satisfy these minimum principles are shown to have a bounded difference with the true scattering states. These minimum principles and stability estimates can be used to obtain rigorous bounds on scattering amplitudes. We show that these minimum principles are applicable to momentum-dependent potentials, long-range (Coulomb) interactions, and elastic or inelastic scattering of bound states.
\end{abstract}
{\renewcommand{\baselinestretch}{0.80}\tableofcontents}
\section{Introduction}
Scattering in a quantum-mechanical system is summarized by the S-matrix. Heuristically, a scattering process begins (in the center-of-mass frame) with a state $|p,-p\rangle$, consisting of two well-separated particles, in approximate momentum eigenstates, one with momentum $p$ and the other with momentum $-p$. Time-evolution maps this to potentially complicated intermediate states in which the particles are interacting, but in the long-time limit we again obtain well-separated particles, now with momenta $k_1,\ldots,k_n$. In the case where the initial states were bound states, they may break up, leading to $n > 2$. The transformation from initial to final states is necessarily unitary, and we call this the action of the S-matrix. The task of computing a scattering process may be phrased as the task of determining elements of the S-matrix.

Computationally, the S-matrix may be accessed without referring to the time-evolution operator at all. A \emph{scattering state} is a Hamiltonian eigenstate with particular asymptotic behavior. In the case of one-body scattering, a scattering state $\psi$ may be decomposed, far from the origin, into a superposition of an ingoing plane wave and many outgoing scattered waves:
\begin{equation}\label{eq:1d-asymptotics}
	\psi(x) = e^{ikx} + \frac{e^{ikr}}{r} f(\theta,\phi) + O(r^{-2})
	\text.
\end{equation}
The amplitudes of the outgoing waves, given by the angle-dependent function $f(\theta,\phi)$, encode elements of the S-matrix.

What do time-independent scattering states have to do with the time-dependent process of scattering? Consider the scattering of a Gaussian wavepacket, of width $\sigma$ and mean-position $x_0$, off of a potential localized near the origin. For a sufficiently broad packet, the width $\sigma$ can be treated as constant, and so we will neglect its time-dependence. The position changes according to $x_0(t) = x_0(0) + t \frac k M$, where $M$ is the mass of the particle. At both early and late times, we have $x_0 \gg \sigma$, and the wavepacket is not interacting with the potential. At intermediate times, the wavepacket does interact with the potential, but for a sufficiently broad wavepacket there are also substantial regions where the evolution of the wavepacket is given by the free evolution of (ingoing and outgoing) plane waves. Therefore, at positions $|x| \ll \sigma$, the wavepacket approximates the time-independent scattering state.

This paper is concerned not with the scattering states themselves, but rather with approximations to scattering states. Let $\tilde\psi$ be a state which satisfies the Schr\"odinger equation only approximately, but has the correct asymptotic behavior for a scattering state: in the case of scattering from a central potential, let $\tilde\psi$ have the asymptotic behavior of \eqref{1d-asymptotics}. Given such a state, we can read off an approximation to the S-matrix from the behavior of the outgoing wave (that is, the function $f(\cdot)$). There are then two natural questions to ask about this approximation to the S-matrix:
\begin{enumerate}
	\item Is this approximation guaranteed to converge to the true S-matrix as $\tilde\psi$ obeys the Schr\"odinger equation more precisely?
	\item Is there a way in which the ``size'' of the violation of Schr\"odinger's equation can be used to set a bound on the size of the error in this approximated S-matrix?
\end{enumerate}
This paper answers both questions affirmatively.

We begin by presenting a minimum principle for scattering states in \secref{states}, broadly similar to Kohn's variational principle~\cite{kohn1948variational}. We show for a wide class of potentials (including all bounded potentials), that the error on the extracted S-matrix can be linearly related to the $L_1$ norm of the Schrodinger violation $(E-\hat H)\tilde\psi$.

Our motivation for this minimum principle is that it provides a computational tool for approximating scattering cross sections in many-body systems. Modern computational methods, including neural quantum states~\cite{Carleo:2016svm}, allow the algorithmic optimization of a wide variety of functionals over the space of wavefunctions. The variational principle we describe identifies such a functional $L[\tilde\psi]$, and accompanies it with the following guarantee: as $L[\tilde\psi]$ approaches $0$, the asymptotic behavior of $\tilde\psi$ will approach that of the true scattering state. Moreover, denoting (some measure of) the error $E$, there exists a constant $C$ such that $E \le C L[\tilde\psi]$ for any $\tilde\psi$.

The coefficient $C$ that characterizes the stability of the scattering state depends on the precise choice of the potential. We discuss this coefficient in~\secref{stability}. For narrow classes of potentials we are able to provide estimates of $C$. In general we are unable to show that any universal bound exists. When available, such estimates make the results of numerical optimization of $L[\tilde\psi]$ more immediately useful. Instead of a sequence of approximations which is only known to be asymptotically correct, each individual approximation is now translated into trustworthy upper and lower bounds on the scattering amplitude---and we still have the guarantee that these bounds will become arbitrarily tight as the approximations are improved. This is similar to the bounds on ground states and real-time dynamics that are provided by various incarnations of the quantum-mechanical bootstrap~\cite{Han:2020bkb,Berenstein:2021dyf,Lawrence:2024mnj,Laliberte:2025xvk}, although of course the methods differ radically.

The remainder of this paper is dedicated to various extensions of these key results. Motivated by common potentials used in nuclear theory, we address momentum-dependent terms in \secref{momentum}, and the Coulomb potential in \secref{coulomb}. We show how to generalize to many-body scattering states (the scattering of two bound states) in \secref{many}, and in the same section discuss the case where bound-state wavefunctions are themselves only approximately known. Finally, in \secref{discussion} we conclude with a discussion of the applications of these results.

\subsection{Prior work}\label{sec:prior}

Variational principles, and related upper and lower bounds, for scattering amplitudes have been considered before. A comprehensive review is not possible here, and we only summarize a few of the more closely related approaches that have been taken, in order to give the reader a sense of the relation between this work and prior approaches. Another partial review is available in~\cite{moiseiwitsch2004variational}. Our paper does not assume any familiarity with this prior work (much of which has remained rather obscure); the historically indifferent reader may safely skip this discussion.

Variational principles abound, including for scattering problems, and continue to be of use in modern calculations. The Lippmann-Schwinger equations come from one such variational principle~\cite{Lippmann:1950zz}, as does Kohn's variational method~\cite{kohn1948variational}. Their numerical application is difficult. Many variational methods suffer from convergence problems, often attributed to the approximation of an infinite-dimensional Hilbert space (that of wavefunctions on a noncompact space) by a finite-dimensional Hilbert space suitable for computational purposes. These convergence problems are also related to the choice of function being minimized. As explored in \appref{counterexamples}, several apparently natural variational principles (including those based on the $L_2$ norm of the Schr\"odinger violation) are not stable, in the sense that the variational principle can be approximately satisfied without the asymptotic behavior of the scattering state being approximately correct.

These convergence difficulties are naturally fixed (as several authors~\cite{spruch1958bounds,shimamura1968variational,bardsley1972minimum} have pointed out) when the variational principle is accompanied by a method for obtaining rigorous upper and lower bounds. It is typically the case that these bounds can be seen to approach each other as the variational state approaches the true state, and therefore we are assured not only that our answer is trustworthy, but also that more computation will eventually yield any desired precision.

Variational principles tied to rigorous upper and lower bounds on amplitudes are less commonly used, but a few have been developed. We will review them here in rough chronological order, and without any claim to be comprehensive. The first such bounds are due to Kato~\cite{kato1951upper}, and (somewhat simplified) begin by minimizing a weighted $L_2$ norm:
\begin{equation}\label{eq:l2-weighted}
	L_2[\tilde\psi; \rho] := \int |(E - \hat H)\tilde\psi|^2 \rho^{-1}
	\text.
\end{equation}
Here $\rho$ is a somewhat arbitrarily chosen weight function which serves to alleviate the convergence problems discussed in the appendix. In order to relate this quantity to bounds on scattering amplitudes, Kato's method requires the user to solve---or at least rigorously bound the spectrum of---a related eigenvalue problem. Later authors alleviated the need for such a solution at the cost of dramatically weakened bounds~\cite{spruch1958bounds}, but also pointed out that there appears to be no general method for selecting a suitable $\rho$~\cite{shimamura1968variational}. This family of methods was also limited in practice to scattering from a single central potential, and equivalent scattering problems. The \emph{least-norm method}~\cite{kanellopoulos1975least} provides an extension of Kato's method to the case of multiple channels.

The minimum-variance method of Bardsley et al~\cite{bardsley1972minimum} continues to consider the weighted $L_2$ norm \eqref{l2-weighted} as a figure of merit for variational states, albeit with the suggestion that it be approximated, in essence, by Monte Carlo integration. However, the minimum-variance method departs from Kato's approach when deriving the bounds. The error in phase shifts was shown by Kato to be given (in the context of scattering from a central potential) by
\begin{equation}
	\Delta_\eta := k\left|\cot \eta - \cot \tilde\eta\right|
	=
	\left|\int_0^\infty \psi (E - \hat H) \tilde \psi \,d r\right|
	\text,
\end{equation}
where $k$ is the incident wavenumber, $\eta$ the true phaseshift, and $\tilde\eta$ the approximated phaseshift from $\tilde\psi$. Bardsley et al point out that, by the Schwarz inequality, this implies the bound
\begin{equation}
	\Delta_\eta
	\le \left(\max_r |\psi(r)|\right)
	\left[\left(\int_0^\infty \rho(r) dr\right)
	\left(\int_0^\infty \rho^{-1}(r)\left| (E - \hat H)\tilde\psi\right|^2dr\right)
	\right]^{1/2}
	\text.
\end{equation}
for any choice of strictly positive weight function $\rho(r)$.

This approach does not require solving an alternative eigenvalue problem---a substantial improvement in ease-of-use over Kato's method---but does introduce the need to calculate, or at least estimate, the maximum magnitude of the true scattering wavefunction. Such an estimate is also needed by our work; unfortunately in~\cite{bardsley1972minimum} such an estimate is provided only for potentials which are sufficiently weak.

We conclude with the works of Darewych, Pooran, and Sokoloff~\cite{darewych1978estimating,darewych1979first,darewych1980upper}, which further improve on the minimum-variance method in a few ways. The first work~\cite{darewych1978estimating} takes the step of relying more heavily on the $L_1$ norm itself, which is the focus of the current work as well. Still restricted to s-wave scattering from a central potential, this is
\begin{equation}
	L_1[\tilde\psi] = \int_0^\infty |(E - \hat H)\tilde\psi|dr
	\text.
\end{equation}
In terms of the $L_1$ norm, the authors reduce the task of obtaining rigorous upper and lower bounds to the phase shift, to the problem of estimating
\begin{equation}
	E_r := \int_0^\infty (\tilde\psi - \psi) (E - \hat H)\tilde\psi\,dr
	\le \left(\max_r \left|\tilde\psi(r) - \psi(r)\right|\right) L_1[\tilde\psi]
	\text.
\end{equation}
As in the work of Bardsley et al~\cite{bardsley1972minimum}, the bound then obtained on $\max |\tilde\psi - \psi|$ only holds for sufficiently weak potentials. This restriction is lifted in~\cite{darewych1979first}, but only by a seminumerical procedure specific to one-dimensional problems: namely, the user of the method must exactly solve the case of a nearby, piecewise-constant potential. And finally, in~\cite{darewych1980upper}, there appear formal (albeit impractical) bounds for the case of multiple dynamical particles.

In the present work we will make use exclusively of the $L_1$ norm. As previous authors did, we find that estimates of the maximum amplitude of the true scattering state are critical. Our emphasis is on providing the basis for a computational method that works for any potential, and will generalize to many-body scattering problems (with inelastic nuclear scattering as a central example). With that in mind, \secref{stability} discusses the circumstances under which $V$-independent bounds on the $L_1$ norm are available. Later sections generalize the method to a wide variety of potentials and scattering problems for which these approaches had not been previously considered.

We will assume that the $L_1$ norm may be satisfactorily approximated by Monte Carlo methods. That is, whereas previous authors were often concerned with being able to rigorously evaluate (or at least bound) the $L_1$ of the Schr\"odinger norm, we will limit ourselves to obtaining bounds on the scattering amplitudes \emph{under the assumption} that trustworthy information about the $L_1$ norm is available. We anticipate that such bounds will be useful to variational Monte Carlo algorithms (similar to~\cite{Flores:2022foz}), which minimize the $L_1$ norm by stochastic gradient descent or a related algorithm~\cite{kingma2014adam}. Taking this perspective allows us to consider methods useful to many-body systems.

\subsection{Notation and conventions}\label{sec:notation}
In line with the convention most common in physics, the complex conjugate of a complex c-number $z$ is written $z^*$; the Hermitian adjoint of an operator $A$ is written $A^\dagger$. Operators on a Hilbert space are denoted with a hat---for example a Hamiltonian will be written $\hat H$.

The $L_p$ norm of a function $f$ is denoted $\|f\|_p$, and is defined for $p=\infty$ as $\|f\|_\infty = \sup_x |f(x)|$. The $L_2$ norm, or magnitude, of a spatial vector $v$ will be denoted $|v|$ or, where unambiguous, simply $v$. When considering scattering from a central potential we will refer to the radial coordinate as $r = |x|$. Therefore the asymptotic behavior of a scattered s-wave is written $\frac{e^{ikr}}{r}$, in lieu of the more explicit but less readabable $\frac{e^{i|k||x|}}{|x|}$. Similarly the unit vector with the same direction as $x$ is written either $\frac{x}{|x|}$ or $\frac{x}{r}$.

To avoid a proliferation of diacritics, we do not generally expressly indicate which objects are vectors (or otherwise have multiple components). The $L_p$ norms of vector-valued functions are defined by
\begin{equation}
	\|f\|_p = \left[\int \sum_i (f_i)^p\right]^{1/p}\text.
\end{equation}

In this work wavefunctions are canonically normalized so that the ingoing plane wave is precisely $e^{ikx}$. Therefore, wavefunctions are dimensionless, unless an ingoing state is a bound state, in which case the wavefunction inherits the dimension of the bound-state wavefunction. We set $\hbar$ to unity throughout, so that momentum has units of inverse length. The $L_1$ norm of the Schr\"odinger violation of a state $\psi$, written $\|(E-\hat H)\psi\|_1$, will appear frequently, and has (in three spatial dimensions) units of length times inverse mass.

\section{Approximate scattering states}\label{sec:states}

To begin with, we consider scattering of a single particle, with no internal degrees of freedom, from a central potential in $3$ spatial dimensions. The potential $V(x)$ is assumed to decay sufficiently rapidly, which we take to mean $V(x) = O(e^{-r/R})$ for some distance $R$, although this constraint can be weakened to a dimension-dependent power-law decay. The Hamiltonian of this system is
\begin{equation}
	\hat H = -\frac{1}{2M} \hat \nabla^2 + V(\hat x)\text,
\end{equation}
where $M$ gives the mass of the particle.

For any incident momentum $k$ and outgoing momentum $p$, the S-matrix element $\langle p|\hat S|k\rangle$ encodes the amplitude for an incident particle of momentum $k$ to be scattered into momentum $p$. Due to energy conservation\footnote{For spherically symmetric potentials, angular momentum is conserved by the scattering process, and so the S-matrix is conventionally expressed in terms of partial waves, as it is diagonal in that basis. The discussion here is more general and will not assume any particular symmetry.}, the S-matrix vanishes unless $|p| = |k|$. A common computational tool for obtaining the (nonrelativistic) S-matrix is the idea of a \emph{scattering state}.

A general one-body scattering state is an eigenstate which obeys particular boundary conditions\footnote{These boundary conditions are more commonly expressed as the Sommerfeld radiation conditions. One removes the ingoing wavefunction entirely, demanding that the remaining wavefunction $u$ satisfy an inhomogeneous Schr\"odinger equation forced by $\hat V \psi_{\mathrm{in}}$. The boundary condition which selects outgoing radiation is~\cite{sommerfeld1912greensche}:
\begin{equation}
	0 = \lim_{\mathrm r\rightarrow \infty} r \left(\frac{\partial u}{\partial r} - ik u\right)
	\text.
\end{equation}}. For our purposes it is convenient to express these boundary conditions by decomposing the wavefunction into three terms: the ``in'' state with the incident plane wave, a ``bulk'' wavefunction which decays rapidly far from the origin, and an ``out'' state which encodes one column of the S-matrix. For an exact scattering state $\psi(x)$ we write
\begin{equation}\label{eq:ss-decomposition}
	\psi(x) = \psi_{\mathrm{in}}(x)
	+ \psi_{\mathrm{bulk}}(x)
	+ \psi_{\mathrm{out}}(x)\text.
\end{equation}
For an incident momentum $k$, the ingoing state is taken to be $\psi_{\mathrm{in}}(x) = e^{i k x}$. The bulk state is assumed to be $O(r^{-2})$. Finally, the outgoing state has the form
\begin{equation}
	\psi_{\mathrm{out}}(x)
	= \frac{e^{ikr}}{r} f\Big(\frac{x}{r}\Big)
	+ O(r^{-2})
	\text.
\end{equation}
Sufficiently far from the origin, this can be approximated as a plane wave in the direction $x / r$. It is an eigenstate (with energy $k^2 / 2 M$) of the Hamiltonian up to terms that are $O(r^{-2})$. Despite the names, not all outgoing behavior is captured by $\psi_{\mathrm{out}}$: the plane wave of $\psi_{\mathrm{in}}$ is both ingoing and outgoing. Note also that this decomposition does not hold for sufficiently long-range potentials, and in particular the Coulomb potential. We postpone the discussion of these to \secref{coulomb}.

This decomposition is not unique: any function which decays sufficiently quickly with $r$ may be added to $\psi_{\mathrm{bulk}}$ and subtracted from $\psi_{\mathrm{out}}$. It can be made unique by taking the $O(r^{-2})$ terms in $\psi_{\mathrm{out}}$ to vanish, but at the cost of making $\psi_{\mathrm{out}}$ singular at the origin. Instead, we choose to leave the decomposition ambiguous, and remain aware that the subleading terms of $\psi_{\mathrm{out}}$ will never appear except in a sum with $\psi_{\mathrm{bulk}}$.

One column of the S-matrix is encoded in $f(\cdot)$. To be precise, assuming the scattering state has incident momentum $k$, then for every momentum $p$ we have
\begin{equation}
	\langle p | \hat S | k \rangle = \delta(p-k) + f\Big(\frac{p}{|p|}\Big) \delta(|p|-|k|)
	\text.
\end{equation}

In order for this to be true, it must be the case that the incident momentum uniquely identifies the scattering state, or at least, that the incident momentum uniquely identifies the asymptotic function $f(\cdot)$. Precisely this is true: as long as $(E-\hat H)\psi = 0$, the outgoing waves are uniquely determined by the ingoing wave.
\begin{theorem}[Uniqueness of scattering]\label{thm:uniqueness}
	Let $\psi^{(1)}(x)$ and $\psi^{(2)}(x)$ be two states, each decomposable in the form \eqref{ss-decomposition}, and each an eigenstate of $\hat H = -\hat \nabla^2 / 2 M + V(\hat x)$ with energy $E$, where $V(x)$ is a local potential. If $\psi^{(1)}_{\mathrm{in}} = \psi^{(2)}_{\mathrm{in}}$, then $\psi^{(1)}_{\mathrm{out}} = \psi^{(2)}_{\mathrm{out}} + O(r^{-2})$.
\end{theorem}
\proof
Since $\psi^{(1)}$ and $\psi^{(2)}$ are both eigenstates of $\hat H$ with energy $E$, their difference $\phi(x) := \psi^{(1)}(x) - \psi^{(2)}(x)$ is as well. Define the probability flux $j_\phi$ (a $3$-component vector field) in this state by
\begin{equation}\label{eq:probability-flux}
	j_\phi(x) = \frac 1 M \Im \phi^* \nabla \phi
	\text.
\end{equation}
Because $\phi$ is an eigenstate of $\hat H$, it obeys $\nabla^2 \phi(x) = 2 M (V(x) - E) \phi(x)$. We then calculate the divergence of $j_\phi$ to find it must vanish everywhere (using the fact that $E$ and $V$ are both real-valued):
\begin{equation}\label{eq:uniqueness-divj}
	\nabla \cdot j_\phi(x)
	= \frac 1 M \Im \phi^*(x) \nabla^2 \phi(x)
	= 2 \Im \phi^*(x) (V(x) - E) \phi(x)
	= 0
	\text.
\end{equation}
The decompositions of $\psi^{(1)}$ and $\psi^{(2)}$ define corresponding asymptotic functions $f^{(1)}$ and $f^{(2)}$. The difference $\phi$ has, at large radii, the form
\begin{equation}
	\phi(x) = \frac{e^{ik r}}{r} f_\phi\Big(\frac{x}{r}\Big) + O(r^{-2})\text.
\end{equation}
From this form we can calculate the integral over a sphere of radius $R$, of the probability flux through the sphere, obtaining
\begin{equation}
	\oint_R d\hat n \cdot j_\phi(x)
	= \int_{S^2} d\Omega\,\frac {|k|} {M} \Big|f_\phi\Big(\frac{x}{r}\Big)\Big|^2 + O(R^{-1})
	\text.
\end{equation}
Applying the divergence theorem to \eqref{uniqueness-divj}, it is immediate that $f_\phi = 0$, or equivalently, $\psi^{(1)}_{\mathrm{out}}$ and $\psi^{(2)}_{\mathrm{out}}$ differ only by terms subleading in $r^{-1}$.
\qed

Of course a stronger statement is true: the entire scattering state is unique, not just the asymptotically large-$r$ behavior. This proof mimics the flavor of many results to come, in which as little information is used about the ``bulk'' wavefunction as possible.

We have assumed that the state $\psi$ is an exact eigenstate of the Hamiltonian, with energy $E = \frac{k^2}{2M}$. The remainder of this paper is concerned with the behavior of states $\tilde\psi$ which are only approximately eigenstates; that is, when the ``Schr\"odinger violation'' $(E-\hat H)\tilde\psi$ is small but non-zero. To be precise, consider an approximate scattering state $\tilde \psi$ and an exact scattering state $\psi$, each decomposed according to \eqref{ss-decomposition}, such that the ingoing waves agree: $\tilde\psi_{\mathrm{in}} = \psi_{\mathrm{in}}$. We assume that the $L_1$ norm of the Schr\"odinger violation is known or bounded:
\begin{equation}
	\int_{\mathbb R^3} |(E - \hat H)\tilde\psi| \le \epsilon
	\text.
\end{equation}
From this statement we will bound the possible error in (the asymptotic behavior of) $\tilde \psi_{\mathrm{out}}$, which is to say the difference between $\tilde \psi_{\mathrm{out}}$ and $\psi_{\mathrm{out}}$. (See \appref{counterexamples} for an explanation of why the $L_1$ norm is a more useful choice than the $L_2$ norm.) Note that the approximate scattering states we will consider are assumed to have exactly the correct ingoing plane wave. The error in these scattering states is only in their failure to be exact eigenstates of the Hamiltonian.

Let us examine what happens to $\nabla \cdot j_\phi$ when the state $\tilde \psi$ is not exactly an eigenstate. In this case, the difference $\phi = \tilde\psi - \psi$ still has no ingoing wave, but also fails to be an exact eigenstate. The divergence is now
\begin{equation}
	\nabla \cdot j_\phi(x) =
	\frac 1 M
	\Im \phi^* \nabla^2 \phi
	= 2 \Im \phi^*(x) (E - \hat H) \phi(x)
	\text.
\end{equation}
Heuristically, this is small because $(E - \hat H)\phi$ is small. However, without a bound on $\phi$ itself, this is not sufficient to show that $j_\phi$ is small. In fact no bound on $\oint j$ follows from the observation that the $L_1$ norm of the Schr\"odinger violation is small. See \appref{counterexamples} for a family of counterexamples demonstrating this. 

To relate the $L_1$ norm of the Schr\"odinger violation to the size of the asymptotic error, we will begin with the following special case of H\"older's inequality.
\begin{lemma}\label{lem:holder}
	Let $f,g$ be complex-valued functions of $\Omega$, with $|g| \le 1$ almost everywhere. Then the $L_1$ norm of $f$ is bounded by:
	\begin{equation}
		\int |f| \ge \left|\int f g\right|
		\text.
	\end{equation}
\end{lemma}
Taking $f = (E - \hat H) \tilde \psi$ and $g$ to be proportional to a particular (exact) scattering state, we obtain a lower bound on the $L_1$ norm, given by the asymptotic behavior of $f$.
\begin{lemma}\label{lem:stable}
	Let $V(x)$ be a short-range potential, $k$ be a momentum, and $\xi$ be any eigenstate of $\hat H = -\frac{\hat \nabla^2}{2M} + V(\hat x)$ with eigenvalue $E = \frac{k^2}{2M}$, for which $\|\xi\|_\infty \equiv \sup_x |\xi(x)|$ is finite. We consider an exact scattering state $\psi(x)$ which obeys $(E-\hat H)\psi = 0$, and a corresponding approximate scattering state $\tilde\psi(x)$, defined to have the respective forms
	\begin{equation}
		\psi(x) = e^{ikx} + \frac{e^{ik|x|}}{|x|} f\left(\frac{x}{|x|}\right) + \psi_{\mathrm{bulk}}(x)
		\text{ and }
		\tilde\psi(x) = e^{ikx} + \frac{e^{ik|x|}}{|x|} \tilde f\left(\frac{x}{|x|}\right) + \tilde \psi_{\mathrm{bulk}}(x)
		\text.
	\end{equation}
	The $L_1$ norm of the Schr\"odinger violation of $\tilde \psi$ is bounded from below by the asymptotic behavior of the difference $\phi(x) \equiv \tilde\psi(x) - \psi(x)$:
	\begin{equation}
		\int |(E-\hat H)\tilde\psi| \,dx
		\ge
		\frac{1}{2 M \|\xi\|_\infty}
		\lim_{R\rightarrow\infty}\left|\oint_R
		\left(\xi^* \nabla_r \phi - \nabla_r \xi^* \phi\right)
		\right|
		\text,
	\end{equation}
	where the surface integral is taken over a sphere of radius $R$.
\end{lemma}
\proof
We begin by noting that, because $\psi$ is an eigenstate, the Schr\"odinger violation of $\tilde\psi$ is precisely that of the difference $\phi$.
Applying H\"older's equality as indicated above, with $\frac{\xi^*}{\|\xi\|_\infty}$ as the test function, the $L_1$ norm of the Schr\"odinger violation is bounded below by
\begin{equation}\label{eq:stab-ineq}
	\int |(E - \hat H)\phi|
	\ge \frac{1}{\|\xi\|_\infty}
	\left|\int \xi^* (E - \hat H) \phi\right|
	\text.
\end{equation}
Although $\xi$ is an eigenstate, the Hamiltonian is not self-adjoint when acting on wavefunctions that are not square-integrable. Integrating by parts, we find that the boundary terms are given by
\begin{equation}\label{eq:stab-eval}
\int \xi^* (E - \hat H) \phi
	= \frac{1}{2M}\int \nabla\cdot (\xi^* \nabla \phi - \nabla \xi^* \phi)
	= \frac{1}{2M}\lim_{R \rightarrow \infty} \oint_{|x| = R} \hat r \cdot \left(\xi^* \nabla \phi - \nabla \xi^* \phi\right)\text.
\end{equation}
Combining \eqref{stab-eval} and \eqref{stab-ineq}, and recalling that $(E-\hat H)\phi = (E - \hat H)\tilde\psi$, the proof is complete.
\qed

Lemma~\ref{lem:stable} is tight, in the sense that for a fixed $L_1$ norm of the Schr\"odinger violation, it is possible to find a wavefunction $\tilde\psi$ which saturates the bound. A worst-case approximate wavefunction $\tilde \psi$ may be constructed by taking $\phi = \tilde\psi - \psi$ to obey
\begin{equation}
	(E - \hat H)\phi \propto \begin{cases}
		\xi & |\xi| = \|\xi\|_\infty\\
		0 & \text{elsewhere.}
	\end{cases}
\end{equation}
In the case when $|\xi| = 1$ only on a set of measure zero, the right-hand side must be proportional to appropriate Dirac delta distributions. This construction is easily seen to saturate the inequality \eqref{stab-ineq}. This is the only inequality used in Lemma~\ref{lem:stable}, which is therefore tight.

On the other hand, for fixed approximate wavefunction $\tilde \psi$, there is no reason to expect that the bound of Lemma~\ref{lem:stable} is saturated. In general one expects that having additional information on $\tilde\psi$, beyond just the $L_1$ norm of the Schr\"odinger violation, may yield a tighter bound on the asymptotic error.

With this lemma in hand, we can establish bounds on different aspects of the approximate scattering state. In the case of scattering from a spherically symmetric central potential, it is conventional to consider partial-wave phase shifts, and indeed these can be constrained by choosing $\xi$ to be the eigenstate in a single partial wave.
\begin{theorem}[Stability of phase shifts in central-potential scattering]\label{thm:stable-phaseshifts}
	Let $V(r)$ be a spherically symmetric potential, and let $\psi$ and $\tilde \psi$ be (respectively) exact and approximate scattering states for an ingoing plane wave $e^{ikx}$, as above. Decomposing the outgoing waves in spherical harmonics according to
	\begin{equation}
		\begin{split}
			\psi(x) -e^{ikx} &= \frac{e^{ikr}}{r} f(\theta) + O(r^{-2}) = \frac{e^{ikr}}{r}\sum_{\ell=0}^\infty (2\ell+1) f_\ell P_\ell(\cos\theta)+ O(r^{-2})
			\text{, and}\\
			\tilde\psi(x) - e^{ikx} &= \frac{e^{ikr}}{r} \tilde f(\theta) + O(r^{-2}) = \frac{e^{ikr}}{r}\sum_{\ell=0}^\infty (2\ell+1) \tilde f_\ell P_\ell(\cos\theta)+  O(r^{-2})
		\text,
		\end{split}
	\end{equation}
	the error in the approximate spherical harmonics $\tilde f_\ell$ is bounded by
	\begin{equation}
		|f_\ell - \tilde f_\ell| \le
		\frac{M}{2 \pi (2\ell+1)} \|\xi_\ell\|_\infty \|(E - \hat H)\tilde\psi\|_1
		\text,
	\end{equation}
	where $\xi_\ell^*$ is the eigenstate of $\hat H$ in partial wave $\ell$.
\end{theorem}
\proof
Labelling the (exponentiated) phase shift $f_\ell$, the asymptotically large-$r$ behavior of $\xi$ may be written\footnote{Note that the definition of $f_\ell$ adopted here does not match the convention followed by (e.g.)~\cite{sakurai2020modern}.}
\begin{equation}
	\xi_\ell^*(r) = \frac{2 \ell + 1}{r} P_\ell(\cos \theta) \left[f_\ell e^{ikr} + \frac{e^{ikr}}{2ik}
-\frac{(-1)^\ell}{2ik} e^{-ikr}
	\right]
	+ O(r^{-2})\text.
\end{equation}
The Schr\"odinger violation of $\tilde\psi$ is the same as the Schr\"odinger violation of $\phi \equiv \tilde\psi - \psi$, the partial-wave decomposition of which reads
\begin{equation}
	\phi(x)
	= \frac{e^{ikr}}{r} f^{(\phi)}(\theta) + O(r^{-2})
	= \frac{e^{ikr}}{r}\sum_{\ell=0}^\infty (2\ell+1) f_\ell^{(\phi)} P_\ell(\cos\theta) + O(r^{-2})\text.
\end{equation}
Here we have defined $f^{(\phi)}_\ell = \tilde f_\ell - f_\ell$, and $f^{(\phi)}(\theta) = \tilde f(\theta) - f(\theta)$.

Invoking Lemma~\ref{lem:stable} and neglecting terms in the integrand which are $O(r^{-3})$ (and therefore do not contribute to the large-$R$ limit), we find the $L_1$ norm of the Schr\"odinger violation is bounded by
\begin{equation}
	\int |(E - \hat H)\tilde \psi|\,dx \ge
	\frac{\left(2 \ell + 1\right)}{2 M\|\xi_\ell\|_\infty} \lim_{R\rightarrow\infty}
	R^{-2}
	\left|
	\oint_R
	f^{(\phi)}(\theta) P_\ell(\cos\theta)
	\right|
	\text.
\end{equation}
Expanding $f^{(\phi)}$ in partial waves, evaluating the azimuthal integral, and exploiting the orthogonality of Legendre polynomials, this reduces to
\begin{equation}
	\int |(E - \hat H)\tilde \psi|\,dx \ge
	\frac{\pi (2\ell+1)^2}{M \|\xi_\ell\|_\infty}
	|f^{(\phi)}_\ell| \int_0^\pi \sin\theta\, d\theta\, \left(P_\ell(\cos\theta)\right)^2
	\text.
\end{equation}
We complete the proof by observing that the Legendre polynomials are normalized such that $\int_{-1}^1 P_\ell^2(u) du = \frac{2}{2\ell+1}$.
\qed

In some contexts it is more interesting to constrain the error in the total cross section. This is somewhat easier to measure in experiment, and it generalizes more readily to the many-body case. In the context of scattering states, the total cross section is obtained by integrating $|f|^2$ over the celestial sphere\footnote{The total cross section may also be obtained by the optical theorem, which relates it to the imaginary part of the forward scattering amplitude. In general approximate scattering states need not satisfy this equality, but the results of this section are strong enough to allow one to bound the amount of violation of the optical theorem.}. To bound the error in this quantity we again use Lemma~\ref{lem:stable}, but with a different choice of test function $\xi$.
\begin{theorem}[Stability of cross sections]\label{thm:stable-crosssection}
	Let $\psi$ and $\tilde \psi$ be exact and approximate scattering states, respectively, with ingoing plane wave $e^{ikx}$ and respective outgoing waves $\frac f {r} e^{ikr}$ and $\frac{\tilde f}{r} e^{ikr}$. From each scattering state define a cross section by integrating over the celestial sphere:
	\begin{equation}
		\sigma = \int d\Omega\,|f|^2
		\quad\text{and}\quad
		\tilde\sigma = \int d\Omega\,|\tilde f|^2
		\text.
	\end{equation}
	If $\psi$ is an exact eigenstate of the Hamiltonian $\hat H = -\frac{\hat \nabla^2}{2M} + V(x)$, then the exact and approximate cross sections are related by
	\begin{equation}
		|\sqrt\sigma - \sqrt{\tilde\sigma}|
		\le \frac{1}{2\sqrt\pi}M \Xi \|(E-\hat H)\tilde\psi\|_1
		\text.
	\end{equation}
	The $\tilde\psi$-independent constant $\Xi$ is given by the supremum of the $L_\infty$ norms of scattering states with ingoing momenta of magnitude $|k|$.
\end{theorem}
\proof
Write the difference in the outgoing behaviors as $f^{(\phi)} \equiv \tilde f - f$. There is an eigenstate $\xi$ which, at large distances $|x|$ from the origin, has the behavior
\begin{equation}
	\xi^* = \frac{e^{-ikr}}{r} f^{(\phi)*} + \frac{e^{ikr}}{r} g + O(r^{-2})\text.
\end{equation}
The function $g$, defined (like $f^{(\phi)}$) on the celestial sphere, is unknown, but will not be needed. We invoke Lemma~\ref{lem:stable} with this as our test eigenstate. The only $r^{-2}$ term in the integrand is $2 i k r^{-2} |f^{(\phi)}|^2$, and so the resulting bound on the $L_1$ norm of the Schr\"odinger violation reads
\begin{equation}\label{eq:xsec-viol-bound}
	\int |(E-\hat H)\tilde\psi|
	\ge
	\frac{k}{M \|\xi\|_\infty} \int d\Omega\,|f^{(\phi)}|^2
	\text.
\end{equation}
The state $\xi$ is linearly determined by the asymptotic error $f^{(\phi)}$, and can be written as an integral over the celestial sphere:
\begin{equation}
	\xi(x) = \frac{|k|}{4\pi}\int_{S^2} d\hat n \,f^{(\phi)}(\hat n) \xi_{\hat n}(x)
\end{equation}
Here the normalization is chosen such that the $\xi_{\hat n}$ are characterized by a canonically normalized ingoing plane wave, $e^{i|k|\hat n x}$.
It follows that the $L_\infty$ norm of $\xi$ is bounded by 
$(4\pi) \|\xi\|_\infty \le |k| \|f^{(\phi)}\|_1 \Xi$, where $\Xi$ is the supremum of the $L_\infty$ norms of $\xi_{\hat n}$. In turn, Cauchy-Schwarz bounds the $L_1$ norm of $f^{(\phi)}$: $\|f^{(\phi)}\|_1 \le \sqrt{4 \pi} \|f^{(\phi)}\|_2$. Applying these two inequalities and \eqref{xsec-viol-bound} we find
\begin{equation}
	\|f^{(\phi)}\|_2 \le \frac 1 {2 \sqrt{\pi}} M \Xi \|(E-\hat H)\tilde\phi\|_1
	\text.
\end{equation}
We complete the proof by invoking the triangle inequality:
$|\sqrt\sigma - \sqrt{\tilde\sigma}| = \big| \|f\|_2 - \|\tilde f\|_2\big| \le \|f - \tilde f\|_2$.
\qed

This bound can be converted into a bound on the difference between cross sections themselves, albeit one which is nonlinear in the $L_1$ norm of the violation.
\begin{corollary}
	Given an approximate scattering state $\tilde\psi$ with corresponding approximated cross section $\tilde\sigma$, the true cross section is bounded by
	\begin{equation}
		\tilde\sigma - 2 \sqrt{\tilde\sigma} C + C^2 \le \sigma \le \tilde\sigma + 2 \sqrt{\tilde\sigma}C + C^2
		\text,
	\end{equation}
	where $C = 2^{-1} \pi^{-1/2} M \Xi \|(E-\hat H)\tilde\psi\|_1$.
\end{corollary}

The general pattern is clear: any eigenstate $\xi$ will result in some bound, relating the $L_1$ norm of the Schr\"odinger violation to the asymptotic behavior of the error $\tilde\psi-\psi$. For one more example, we can extract a pointwise bound on the error of the asymptotic function $\tilde f$. This is done by using an exact scattering state produced by an ingoing plane wave with momentum of arbitrary direction as a test function.
\begin{theorem}
	Let $\psi$ and $\tilde\psi$ be (respectively) exact and approximate scattering states, with ingoing plane wave $e^{ikx}$ and respective outgoing waves $\frac f {r} e^{ikr}$ and $\frac{\tilde f}{r} e^{ikr}$. Select a momentum $p$ such that $|p| = |k|$, and let $\xi$ be the (exact) scattering state produced by the ingoing plane wave $e^{-ipx}$. The pointwise difference between $f$ and $\tilde f$ is bounded by
	\begin{equation}
		|\tilde f - f| \le \frac{M \|\psi\|_\infty}{\pi} \int |(E-\hat H)\tilde\psi|
		\text.
	\end{equation}
\end{theorem}
\proof
We take for our test function the exact scattering state:
\begin{equation}
	\xi^*(x) = e^{ipx} + \frac{e^{ikr}}{r} \left[f\left(\frac{x}{r}\right)
	+ O(r^{-1})\right]
	\text.
\end{equation}
Invoking Lemma~\ref{lem:stable}, the $L_1$ norm of the Schr\"odinger violation is bounded below by:
\begin{equation}
	\int |(E-\hat H) \tilde\psi| \, d^3x
	\ge \frac{1}{2M \|\xi\|_\infty}
	\lim_{R\rightarrow\infty}\left|\oint_R
		ik \frac{e^{ikr(1+\cos\theta)}}{r} f^{(\phi)}\left(\frac{x}{r}\right)(1-\cos\theta)
		\right|
	\text.
\end{equation}
Here $\theta$ is the angle between the momenta $p$ and $k$.
The above expression neglects all terms in the integrand which (after taking the stationary-phase approximation) will integrate to $O(R^{-3})$ and therefore fail to contribute to the limit. One point (at $\cos\theta=-1$) gives a non-zero contribution to the first term:
\begin{equation}
	\begin{split}
	\int |(E-\hat H) \tilde\psi| \, d^3x
		&\ge \frac{|k|}{2M \|\xi\|_\infty}
		\lim_{R\rightarrow\infty} R \left|\oint_{S^2}
	f^{(\phi)}\left(\frac{x}{|x|}\right) e^{i|k|R(1+\cos\theta)} (1-\cos\theta)
	\right|\\
		&= \frac{\pi}{M\|\xi\|_\infty} f^{(\phi)}(\cos\theta=-1)
	\text.
	\end{split}
\end{equation}
Bounding the error in $\tilde f$ at different points requires selecting different scattering states $\xi$ by changing the ingoing momentum $p$.
\qed

In the case of the free theory (i.e.~$V=0$), the exact scattering state $\xi$ is known exactly. In one dimension this is $\xi(x) = e^{ikx}$, and in three dimensions we have $\xi(r) = \frac{\sin k r}{k r}$. These states each obey $|\xi(x)| \le 1$ at all points. As a result we have the following (computationally uninteresting) result regarding approximate eigenstates of $\hat p^2$, which says that a small amount of Schr\"odinger violation can only create a small amount of scattering.
\begin{corollary}
	Let $\tilde\psi(x)$ be a scattering state, with ingoing wave $e^{ikx}$ and outgoing wave $r^{-1} f(x/r) e^{ikr}$. The asymptotic function $f(\cdot)$ is bounded by the $L_1$ norm of the Schr\"odinger violation according to
	\begin{equation}
		|f| \le \frac{M}{\pi} \|(E-\hat H_0)\tilde\psi\|_1
		\text.
	\end{equation}
\end{corollary}

In order for Lemma~\ref{lem:stable} and the subsequent theorems to be of use, we need to know that eigenstates of $\hat H$ have such pointwise bounds. This is a potential-dependent statement: see \appref{counterexamples} for a case of a potential for which there are pointwise-unbounded eigenstates. Nevertheless it holds true for ``reasonable'' potentials, as is revealed by the following heuristic argument.

Suppose $\psi$ obeys $(E + \nabla^2 / 2 M - V(x) ) \psi(x) = 0$. Let the leading-order behavior of the wavefunction, at small $x$, be $\psi = x^\alpha + o(x^\alpha)$. If $V$ is finite and smooth at $x=0$, it is easy to see that the only solutions to the eigenvalue equation are $\alpha = 0,1,2$. In particular, no power-law divergence of $\psi$ with $x$ (at small $x$) can be an eigenstate.

Specific pointwise bounds will be discussed in the next section. For now we must be content with the fact that eigenstates have pointwise bounds for potentials which are pointwise bounded and finite range.
\begin{lemma}\label{lem:bounded} Let $V(x)$ be a real-valued function of $\mathbb R^N$ obeying $|V(x)| \le V_0$ for some finite $V_0$, and supported only on $|x| \le R$. Then for any real energy $E > 0$, there exists a constant $B_{E,V}$ such that any eigenstate $\psi$ with eigenvalue $E$ of the Hamiltonian
		$\hat H = -\frac{\hat \nabla^2}{2M} + V(\hat x)$
	has magnitude bounded pointwise by $B_{E,V}$: $|\psi(x)| \le B_{E,V}$.
\end{lemma}
A proof for this result, follows from the standard theory of elliptic regularity (establishing that $\psi$ has smooth derivatives of arbitrarily high order) and Sobolev inequalities~\cite{evans2022partial}.

Putting this together, we find that the asymptotic behavior of scattering states is stable relative to the $L_1$ norm of the Schr\"odinger violation. Absent a determination of the coefficient of stability, the key takeaway is summarized as follows:
\begin{theorem}[Stability of scattering]\label{thm:stable}
	Let $V(x)$ be a local (i.e.~exponentially decaying) potential, and $\hat H$ the corresponding Hamiltonian for a particle of mass $M$. There exists a sequence of finite constants $C_\ell$ such that, for any approximate scattering state $\tilde\psi$ with asymptotic outgoing behavior $r^{-1} e^{ikr} \tilde f(x/|x|)$ and an ingoing behavior $e^{ikx}$, the inferred partial wave amplitudes have error bounded by
	\begin{equation}
		|f_\ell - \tilde f_\ell| \le C_\ell \|(E-\hat H)\tilde \psi\|_1
		\text.
	\end{equation}
	Similarly there exists a finite constant $D$ such that the approximated cross section obeys
	\begin{equation}
		|\sqrt{\tilde \sigma} - \sqrt{\sigma}| \le D \|(E-\hat H)\tilde \psi\|_1
		\text.
	\end{equation}
\end{theorem}

In this section we have limited discussion to the case of a single particle interacting with a central potential, but by the usual reduced-mass formalism, this is equivalent to considering two particles interacting with a potential which depends only on the distance between them. We will generalize this to the case of several particles interacting simultaneously in \secref{many}.

\section{Stability estimates}\label{sec:stability}
Lemma~\ref{lem:stable} requires that a pointwise bound on the magnitude of the eigenstate be known, and while Lemma~\ref{lem:bounded} guarantees that such a bound exists, it provides no guidance on the size of $\max |\xi|$. Furthermore, Lemma~\ref{lem:bounded} is established for each energy and potential separately. It does not guarantee that any single such bound holds for an entire class of potentials. For example one might conjecture that all bounded potentials of bounded support have scattering states respecting the same $L_\infty$ bound:
\begin{conjecture} Let $V$ obey $|V| \le V_0$, and $V(x) = 0$ when $|x| \ge R_0$, for some finite $V_0,R_0$. Then any scattering state $\psi$ of $\hat H = \frac{\hat p^2}{2M} + V(\hat x)$ obeys $|\psi| \le B_{V_0,R_0}$, for a universal constant $B_{V_0,R_0}$ which is independent of the details of $V$ (but may depend on the incident momentum and mass).
\end{conjecture}

This conjecture is phrased in terms of the $L_\infty$ norm of a scattering state, which is not a ``physical'' quantity: there is no experiment which directly measures it. However, if true, this conjecture immediately implies the existence of a universal bound on the cross section of $V_0,R_0$-bounded potentials. Similarly, as both the cross section and the $L_\infty$ norm $\|\psi\|_\infty$ blow up near a resonance, this conjecture would imply a maximum lifetime of unstable states in bounded potentials. This latter consequence is easily seen to be true in the semiclassical limit: the WKB tunneling rate is only exponential in $V_0$ and $R_0$.

To our knowledge, this conjecture is neither proven nor refuted in three dimensions. In one dimension such bounds are easily obtained for local and bounded potentials.
\begin{lemma}\label{lem:1dbound}
	Let $V(x)$ be a potential bounded by $|V(x)| \le V_0$, and which vanishes outside the interval $[-L/2,L/2]$, and let $\psi$ be a real-energy eigenstate, of positive energy $E$, of the corresponding Hamiltonian $\hat H = \frac{\hat p^2}{2M} + V(\hat x)$. Furthermore require that on the right half-line $[L/2,\infty)$, the wavefunction is $\psi = A e^{ikx}$ for some $|A| \le 1$. Then the eigenstate $\psi$ satisfies the pointwise estimate
	\begin{equation}
		|\psi| \le \sqrt{1 + k^2} e^{L\max\{2M(V_0 + E),1\}}
		\text.
	\end{equation}
\end{lemma}
\proof
The one-dimensional Schr\"odinger equation is a second-order ODE, and therefore is equivalent to the following two-dimensional first-order ODE:
\begin{equation}
	\frac{d}{dx} f \equiv
	\frac{d}{dx}
	\left(\begin{matrix}\psi\\\psi'\end{matrix}\right)
		=
		\left(
		\begin{matrix}
			0 & 1 \\
			2M(V(x) - E) & 0
		\end{matrix}
		\right)
	\left(\begin{matrix}\psi\\\psi'\end{matrix}\right)
		\equiv
		A f
	\text.
\end{equation}
Here we have labelled the two-component vector of $\psi$ and $\psi'$ by $f$. The solution is defined by the boundary conditions at $x = \frac L 2$. For any $C$ greater than or equal to the spectral norm of $A$ ($C \ge \|A\|$), the solution is bounded everywhere in the interval by
\begin{equation}
	|\psi| \le \|f\| \le e^{C L} \|f(x=L/2)\|\text.
\end{equation}
Evaluating $\|A\| \le \max \{1,2M(V_0 + E)\}$, and noting that $\|f\| \le \sqrt{1 + k^2}$ at $x=\frac L 2$, completes the proof.
\qed

As mentioned in \secref{prior}, Bardsley et al obtained a bound for sufficiently weak potentials via the Lippmann-Schwinger equation~\cite{bardsley1972minimum}, which we reproduce here for completeness, specialized to the case of s-wave scattering in three dimensions.
\begin{theorem}Let $\psi(r)$ be an s-wave scattering state from a central potential $V(r)$. Then the amplitude $|\psi(r)|$ is bounded by $\frac 1 {1 - \gamma_*}$ whenever $\gamma_* < 1$, where
	\begin{equation}
	\gamma_* := \max_r \int dr'\,|G_0(r,r')V(r')|
		\text.
	\end{equation}
\end{theorem}
\proof
Via Lippmann-Schwinger, the s-wave scattering state $\psi(r)$ obeys
\begin{equation}
	\psi(r)
	= \psi_0(r)
	+ \int G_0(r,r') V(r') \psi(r')\,dr
	\text.
\end{equation}
Here $G_0$ and $\psi_0$ are the Green's function and scattering state in the absence of the potential, respectively; namely
\begin{equation}
	G_0(r,r') =
	\frac{-M}{2\pi} \frac{e^{ik|r'-r|}}{|r'-r|}
	\text{ and }
	\psi_0(r) = \sin k r
	\text.
\end{equation}
Define $r_*$ to be the coordinate at which $|\psi(r)|$ is maximized. Applying the Lippmann-Schwinger equation at $r_*$ we establish the chain of inequalities
\begin{equation}
	\begin{aligned}
	\max_r |\psi(r)|
	= |\psi(r_*)|
		&\le |\sin k r_*| + \left|\int dr'\, G_0(r_*, r') V(r') \psi(r')\right|\\
		&\le 1 + |\psi(r_*)| \int dr'\, |G_0(r_*, r') V(r')| \\
		&\le 1 + \left[\max_r |\psi(r)|\right] \max_r \int dr'\,|G_0(r,r') V(r')|
	\text.
	\end{aligned}
\end{equation}
Defining $\gamma_* = \max_r \int dr'\,|G_0(r,r')V(r')|$, it follows that the maximum amplitude of the wavefunction is bounded according to
\begin{equation}
	\max_r |\psi(r)| \le \frac{1}{1 - \gamma_*}
\end{equation}
whenever $\gamma_* < 1$.
\qed

Neither of the above two bounds is much use in general. The first requires the problem to be effectively one-dimensional: it can be applied to three-dimensional scattering only in the spherically symmetric case, by reducing to the radial Schr\"odinger equation. The second requires that the potential be sufficiently weak (and moreover requires the difficult computation of $\gamma_*$ in order to be useful). Obtaining useful pointwise bounds on positive-energy eigenstates for reasonably general potentials remains, as far as we know, an open problem. 

\section{General potentials}\label{sec:potentials}
Realistic physical Hamiltonians are not as limited as those considered in the sections above. The purpose of this section is to show that the results of \secref{states} apply to such Hamiltonians without significant modification. In the first subsection below we consider both internal degrees of freedom and momentum-dependent potentials; in the second subsection below we consider long-range Coulomb interactions.

\subsection{Momentum and spin dependence}\label{sec:momentum}

A central motivation for this work is the nuclear scattering problem, considered in a nuclear effective theory. In nuclear effective theories the mesons have been integrated out to obtain non-local, but short-range, forces between nucleons. An early example of high-quality nuclear forces, capable of accurately reconstructing two-nucleon scattering data as well as light nuclear binding energies, are the Argonne family of potentials~\cite{Wiringa:1994wb,Pudliner:1997ck,Wiringa:2002ja}. More recent works typically use forces based on chiral effective theories: see~\cite{Machleidt:2011zz} for a review.

Common to all realistic nuclear potentials is that the potential has non-trivial dependence on both momentum and on internal degrees of freedom (e.g.~spin and isospin). As a motivating example consider the following one-body Hamiltonian, reminiscent of the interaction between two nucleons:
\begin{equation}
	\hat H = \frac{\hat p^2}{2M} + e^{-|x|^2/2} \left(g_1 \hat L \cdot \hat \sigma+ g_2 \hat L^2\right)\text.
\end{equation}
The wavefunctions on which this Hamiltonian acts have two complex components for each position; that is, $\psi_s(x)$ is a complex number for any $x \in \mathbb R^3$ and $s \in \{0,1\}$. The index represents spin and is acted on by the three Pauli matrices $\sigma_x,\sigma_y,\sigma_z$.

The potential---here consisting of all terms excepting the one proportional to $\hat p^2$---is not diagonal in position space, but is also not an entirely arbitrary operator. The conservation laws have not been affected: angular momentum is still conserved, and in the full two-body problem, momentum and the center of mass are both conserved. Additionally, the potential is short range, falling off rapidly with distance from the origin (or distance between the two particles). In fact, if we neglect electromagnetic effects (to be addressed in the next subsection), the nuclear potential decays exponentially with distance.

Because the potential is short-range, the asymptotic structure of a scattering state is not modified by the inclusion of momentum-dependent terms. Additionally, although the kinetic term $\propto \hat p^2$ is not self-adjoint on the vector space of scattering states, exponentially decaying momentum-dependent potentials still are. For example, for any wavefunctions $f$ and $g$ growing at most polynomially with $|x|$:
\begin{equation}
	\int f^*(x) e^{-|x|} \hat L^2 g(x)
	= \int \left(\hat L^2 f^*(x)\right) e^{-|x|} g(x)
	\text.
\end{equation}
Boundary terms do not cause a failure of short-range operators to be self-adjoint: an operation ``naively'' self-adjoint on the space of square-integrable wavefunctions is also self-adjoint on scattering states. This observation will be repeatedly used below.

For arbitrary momentum-dependent potentials, the probability current as defined by \eqref{probability-flux} does not always obey the continuity equation, which reads:
\begin{equation}
	\nabla \cdot j_\psi(x) = \frac 1 M \Im \psi^*(x) \hat H \psi(x)
	\text.
\end{equation}
A suitable probability current can be found, but it suffices here to note that while $j$ as defined by \eqref{probability-flux} is not locally conserved, due to the rapid decay of the potential with distance, it still obeys $\int \nabla \cdot j = 0$ in an eigenstate. Therefore scattering remains unique when the potential is momentum-dependent. The proof closely follows the original proof in Theorem~\ref{thm:uniqueness}, and either choice of current suffices as long as the potential decays rapidly at large distances.
\begin{theorem}[Uniqueness of momentum-dependent scattering]
	Let $\psi^{(1)}$ and $\psi^{(2)}$ be $k$-component scattering states, with ingoing wave $u e^{ikx}$ for some complex $k$-vector $u$ obeying $|u| = 1$.
	Define the Hamiltonian $\hat H = \frac{\hat p^2}{2M} + \hat V$, where the potential $\hat V$ is self-adjoint on wavefunctions which are $O(|x|^{-1})$. If each scattering state obeys $(E - \hat H)\psi^{(a)} = 0$, then $\psi^{(1)}_{\mathrm{out}} = \psi^{(2)}_{\mathrm{out}}$.
\end{theorem}
\proof
Consider the difference $\phi(x) := \psi^{(1)}(x) - \psi^{(2)}(x)$, which is an eigenstate of $\hat H$ with energy $E$. We will use the original probability flux $j_\phi$, defined by
\begin{equation}
	j_\phi(x) = \frac 1 M \Im \phi^* \nabla \phi
	\text.
\end{equation}
We now compute the divergence of $j_\phi$, using the fact that $\phi$ obeys the eigenstate equation $\nabla^2 \phi(x) = 2 M (\hat V - E) \phi(x)$:
\begin{equation}
	\nabla \cdot j_\phi(x)
	=
	\frac 1 M \Im \phi^*(x) \left[2 M (\hat V - E) \right]\phi(x)
	= 2 \Im \phi^*(x) \hat V \phi(x)
\end{equation}
As discussed, the probability flux does not vanish identically. However, as $\hat V$ is a self-adjoint operator over a vector space of which $\phi$ is an element, the integral of this divergence does vanish:
\begin{equation}
	\Im \int \phi^*(x) \hat V \phi(x)
	= \Im \int \hat V \phi^*(x) \phi(x)
	= \Im \left(\int \phi^*(x) \hat V \phi(x)\right)^*
	= 0
	\text.
\end{equation}
Then as in the momentum-independent case, the decompositions of $\psi^{(1)}$ and $\psi^{(2)}$ define corresponding asymptotic functions $f^{(1)}$ and $f^{(2)}$. The difference $\phi$ has the form
\begin{equation}
	\phi(x) = \frac{e^{ikr}}{r} f\big(\frac{x}{r}\big) + O(r^{-2})\text.
\end{equation}
Integrating the probability flux through a sphere of radius $R$, we find
\begin{equation}
	\oint_R d\hat n \cdot j_\phi(x)
	= 4 \pi \int d\Omega\,\frac {|k|} {M} \Big|f\big(\frac{x}{|x|}\big)\Big|^2 + O(R^{-1})
	\text.
\end{equation}
It is immediate that $f = 0$, and therefore $f^{(1)} = f^{(2)}$.
\qed

Lemma~\ref{lem:stable} also applies to the case of a momentum-dependent local potential with no modification needed in the proof. This is again a result of the fact that $\hat V$ is self-adjoint. In the case of internal degrees of freedom---where the wavefunction has multiple components---it is worth restating the lemma, as there are multiple possible notions of the various norms that appear.

On the space of $n$-component wavefunctions $\mathbb R^3 \rightarrow \mathbb C^n$, we define $L_1$ and $L_\infty$ norms by
\begin{equation}
	\|\psi\|_1 = \int
	\sum_i |\psi_i|
	\quad\text{and}\quad
	\|\psi\|_\infty =
	\max_{i,x} |\psi_i(x)|
	\text.
\end{equation}
These definitions blend nicely with the following multi-component generalization of Lemma~\ref{lem:holder}.
\begin{lemma}
	Let $f_i,g_i$ $n$-component complex-valued functions, with $|g_i| \le 1$ almost everywhere and for all $n$. Then the $L_1$ norm of $f$ is bounded by
	\begin{equation}
		\int \sum_i |f_i|
		\ge
		\left|
		\int
		\sum_i f_i g_i
		\right|
	\end{equation}
\end{lemma}
This is not the only possible generalization of Lemma~\ref{lem:holder}, but it is sufficient to derive a generalization of Lemma~\ref{lem:stable}. The proof proceeds in the same manner as the original proof, and will not be repeated here.
\begin{lemma}
	Let $\hat V$ be a self-adjoint, short-range potential, $k$ be a momentum and $E$ be a real energy. Additionally, let $\xi$ be a (multi-valued) wavefunction which is an eigenstate of $\hat H = -\frac{\hat\nabla^2}{2M} + \hat V$ with eigenvalue $E$, for which $\|\xi\|_\infty$ as defined above is finite. Then, letting $\psi,\tilde\psi$ respectively be an exact and approximate scattering state each with the usual decomposition, the $L_1$ norm of the Schr\"odinger violation of $\tilde \psi$ is bounded by:
	\begin{equation}
		\|(E-\hat H)\tilde\psi\|_1
		\ge
		\frac{1}{2 M \|\xi\|_\infty}
		\lim_{R\rightarrow\infty}\left|\oint_R
		\left(\xi^* \nabla_r \phi - \nabla_r \xi^* \phi\right)
		\right|
		\text.
	\end{equation}
\end{lemma}
From this lemma, all subsequent theorems in \secref{states} follow as before.

\subsection{Coulomb potentials}\label{sec:coulomb}

Not all potentials are short-range. The most common long-range potential is the Coulomb potential. Neglecting spin and relativistic effects, we are now concerned with the Hamiltonian of a single particle in a static electric field (equivalent to the interaction of two charged particles):
\begin{equation}\label{eq:H-coulomb}
	\hat H = -\frac{1}{2M} \nabla^2 + \frac{\alpha}{r}
	\text.
\end{equation}

As written, the only remarkable feature of this potential is that it decays quite slowly---slowly enough that, as we will see later, the asymptotic behavior of scattering states has an entirely different structure. However, the same physics may be expressed by a Hamiltonian of the form
\begin{equation}
	\hat H = -\frac{1}{2M} (\nabla - A)^2
	\text,
\end{equation}
with suitably chosen vector potential $A$.
These two are related by a gauge transformation. When considering this second form, the definition of the probability current must be changed. The new probability current is found by applying the same gauge transformation to the usual definition, and reads:
\begin{equation}
	j_{\mathrm{gauge}} = \frac 1 M \Im \psi^* \nabla \psi - \frac 1 M A |\psi|^2
	\text.
\end{equation}
For the purposes of this work we will stick with the form (\ref{eq:H-coulomb}), where no redefinition of the current is needed: $j = \frac 1 M \Im \psi^* \nabla \psi$.

The long-range nature of the Coulomb interaction changes the behavior of eigenstates even far from the origin.
The $\ell$th partial wave has the form, at leading order in $r^{-1}$:
\begin{equation}
	\psi_\ell(r) =
	\frac{e^{\pm i \theta_\ell}}{r}
	+ O(r^{-2})
	\text{ where }
	\theta_\ell = kr - \frac{M\alpha}{k} \log(2kr) - \frac{\ell\pi}{2}
	+ \arg\Gamma \left(\ell + 1 + i \frac{m\alpha}{k}\right)
\end{equation}
Alternatively, we can examing the asymptotic behavior of an eigenstate with an ingoing plane wave and arbitrary outgoing waves. The nature of the \emph{ingoing} wave must change in order to obtain an eigenstate. At leading order in $r^{-1}$ this wavefunction reads~\cite{Mulherin:1970ej,landau2013quantum}
\begin{equation}\label{eq:ss-coulomb-decomposition}
	\psi(x) = e^{i\left(k\cdot x + \frac {M\alpha}{|k|}\log(|k||x| - k\cdot x)\right)}
	+ \frac{e^{i\left(|k||x| - \frac{M\alpha}{|k|}\log(2 |k||x|)\right)}}{|x|} f\left(\frac{x}{|x|}\right)  + O(|x|^{-2})
	\text.
\end{equation}
The outgoing current in this state is infinite, reflecting the infinite cross section of Coulomb scattering.

\begin{theorem}[Uniqueness of Coulomb scattering]
	Let $\psi^{(1)}(x)$ and $\psi^{(2)}(x)$ be two states, each decomposable in the form \eqref{ss-coulomb-decomposition} and each an eigenstate of $\hat H = \nabla^2 / 2M + V(\hat x) + \alpha/|\hat x|$.
If $\psi^{(1)}_{\mathrm{in}} = \psi^{(2)}_{\mathrm{in}}$, then $\psi^{(1)}_{\mathrm{out}} = \psi^{(2)}_{\mathrm{out}} + O(r^{-2})$ almost everywhere.
\end{theorem}
\proof Consider the difference between the two states. The ingoing behavior cancels, and what remains is
\begin{equation}
	\phi = \tilde\psi - \psi = 
	\frac{e^{i k x - i \frac{n\alpha}{|k|} \log 2 k |x|}}{|x|} \left[\tilde f \left(\frac{x}{|x|}\right) - f\left(\frac{x}{|x|}\right)\right]
	+ O(r^{-2})
	\text.
\end{equation}
Since each of $\psi$ and $\tilde\psi$ is an eigenstate, the difference $\phi$ is as well. In an eigenstate $\nabla \cdot j_\phi = 0$. We evaluate the integral of this divergence to find
\begin{equation}
	0 = \oint_R d\hat n \cdot j_\phi(x) =
	\frac {|k|} M \oint_{S^2} |f_\phi|^2
	+ O(R^{-1})
	\text.
\end{equation}
Taking the limit $R\rightarrow 0$, this integral can only vanish if $f_\phi \equiv \tilde f - f$ almost everywhere, which completes the proof.
\qed

The stability of scattering in an asymptotically Coulomb potential is provided by the following lemma, analogous to Lemma~\ref{lem:stable}.
\begin{lemma}
	Let $V(x)$ be a short-range potential, $k$ be a momentum, $\alpha$ a real number parameterizing the Coulomb interaction, and $\xi$ be any eigenstate of $\hat H = -\frac{\nabla^2}{2M} + V(\hat x) + \frac{\alpha}{|\hat x|}$ with eigenvalue $E$, for which $\|\xi\|_\infty \equiv \sup_x |\xi(x)|$ is finite. We consider an exact scattering state $\psi(x)$ which obeys $(E-\hat H)\psi = 0$, and a corresponding approximate scattering state $\tilde\psi(x)$, each defined to be of the form of Eq.~(\ref{eq:ss-coulomb-decomposition}).
	The $L_1$ norm of the Schr\"odinger violation of $\tilde \psi$ is bounded by the asymptotic behavior of the difference $\phi(x) \equiv \tilde\psi(x) - \psi(x)$:
	\begin{equation}
		\int |(E-\hat H)\tilde\psi| \,dx
		\ge
		\frac{1}{2 M \|\xi\|_\infty}
		\lim_{R\rightarrow\infty}\left|\oint_R
		\left(\xi^* \nabla_r \phi - \nabla_r \xi^* \phi\right)
		\right|
		\text,
	\end{equation}
	where the surface integral is taken over a sphere of radius $R$.
\end{lemma}
The proof of this lemma is exactly as the proof of Lemma~\ref{lem:stable}. From this, the same sorts of theorems result as in the short-range case; however note that we cannot bound the change in the cross section, as the cross section (as usually defined) is infinite. Differential cross sections, and appropriately cut-off total cross sections, can still be bounded. We leave a thorough treatment to future work.

\section{Many-body scattering}\label{sec:many}
In this section we turn our attention to systems of multiple ($>2$) interacting particles. We are concerned with the following sort of scattering event: two bound states interact, and possibly break up into at least three outgoing states, which may be bound or individual particles.

Two-body scattering is, in practice, numerically tractable, in the sense that nearly-exact scattering states may readily be found even for realistic potentials. These states are close enough to exact that there is little practical use for the stability estimates provided above. The same is not true for many-body scattering. In particular, a first-principles calculation of inelastic scattering with three outgoing particles (e.g.~the breakup of a deuteron by an incident nucleon) is already computationally difficult. To the knowledge of these authors, no such calculations are available for strongly coupled, inelastic scattering involving four or more dynamical particles, while all systematics are under control. In short, the general case of inelastic scattering with four or more particles remains an unsolved computational problem.

The case of incident bound states introduces a significant complication into the $L_1$ norm-based method we use in this work. For results obtained up until this point, the asymptotic form of the wavefunction was known \emph{a priori}, and in particular the ingoing wave $e^{ikx}$ was specified exactly. However, when the ingoing states are bound states it can no longer be assumed that their wavefunctions are known exactly. In particular, in any practical calculation, the $L_1$ norm of the scattering state will be infinite, as the small but non-zero error in the ansatz for the bound state will be multiplied by the spatial volume of the calculation. A similar difficulty will occur with any outgoing states that happen to be bound states. These difficulties can be disregarded only if the numerical evaluation of the $L_1$ norm is done ``sloppily'', in such a way as to render negligible the formally infinite contribution from the error in the bound-state wavefunction.

To have a concrete example, consider the scattering of a two-particle bound state off of a central potential. We will consider only energies below the binding energy of the bound state, such that the outgoing state is also guaranteed to be a bound state, and with the same wavefunction. This is therefore an elastic process, but distinct from the problems considered in earlier sections due to the additional dynamical coordinate.

Labelling the bound-state wavefunction $\Psi$, the scattering state for this process may be decomposed in a form analogous to \eqref{ss-decomposition} above:
\begin{equation}
	\psi(x_1, x_2)
	=
	e^{ikx}
	\Psi(x_2 - x_1)
	+
	\frac{e^{ikr}}{r}
	f\left(\frac{x}{r}\right)
	\Psi(x_2 - x_1)
	+ \psi_{\mathrm{bulk}}(x_1,x_2)
	\text.
\end{equation}
Above we have denoted the center-of-mass coordinate (which is not conserved) $x = \frac{x_1 + x_2}{2}$, and $r = |x|$ as usual.
The bulk wavefunction $\psi_{\mathrm{bulk}}$ is required to fall strictly faster than $|x_1|^{-3}$ and $|x_2|^{-3}$.
A sensible ansatz for an approximate scattering state will then take this form, but with variational ans\"atze for the bound-state wavefunction $\Psi$, the bulk behavior $\psi_{\mathrm{bulk}}$, and the asymptotic structure $f$. The $L_1$ norm of such a state will, as mentioned, be infinite. Additionally, there are now two sources of error in the S-matrix element read out from such an approximate scattering state: the error in the ``bulk'' evolution of Schr\"odinger's equation, but also the error in the choice of bound state wavefunction.

Above, the overall asymptotic structure of the scattering state is the same as it is for the case of single-particle scattering (from a central potential). This is no longer the case once we allow for inelastic scattering. In general, accounting for all scattering channels requires a comprehensive list of asymptotic states which are kinematically allowed.

In this section we will first assume that the bound-state wavefunctions are, in fact, known exactly. In the second subsection, we will show how finite error in the bound-state wavefunctions can be folded in, to obtain rigorous estimates on the scattering of approximated bound states. In all cases we will work with Hamiltonians of the form 
\begin{equation}
	\hat H = -\sum_{n=1}^N \frac{1}{2M_n} \nabla^2_n
	+
	\sum_{n_1,n_2} V_{n_1,n_2}(x_{n_1},x_{n_2})
	+
	\sum_{n_1,n_2,n_3} V_{n_1,n_2,n_3}(x_{n_1},x_{n_2},x_{n_3})
	+
	\cdots
	\text.
\end{equation}
Here $V_{n_1,n_2}$ is the two-body potential between particles $n_1$ and $n_2$, $V_{n_1,n_2,n_3}$ is the three-body potential, and so on. It is assumed that $V_\bullet$ is real and decays exponentially with the largest separation among the arguments. For brevity we will summarize these potential terms by a single function $V(x_1,\ldots,x_N)$, deemed a local $N$-body potential precisely when the above decomposition holds.

\subsection{Approximate scattering of exact bound states}\label{sec:exact-approx}
We begin with a concrete example that illustrates the key features of bound-state scattering. Consider a bound state of two distinguishable particles, bound by the effect of an interaction $V_i$, incident on a fixed (short-range) central potential $V_c$, with nonvanishing outgoing amplitudes both for elastic and inelastic scattering. The Hamiltonian of this system is
\begin{equation}\label{eq:H2-inelastic}
	\hat H = \frac{\hat p_1^2}{2M} 
	+ \frac{\hat p_2^2}{2M}
	+ V_i(\hat x_1 - \hat x_2) + V_c(\hat x_1) + V_c(\hat x_2)
	\text.
\end{equation}
The bound-state wavefunction is a complex-valued function of a single vector, the relative coordinate: $\Psi(x_2-x_1)$. Plane waves constructed from this wavefunction are eigenstates of the full Hamiltonian given above. The corresponding scattering state has the form
\begin{equation}
	\psi(x_1,x_2) =
	e^{ik\frac{x_1+x_2}{2}}\Psi(x_1-x_2)
	+
	\psi_{\mathrm{bulk}}(x_1,x_2)
	+
	\psi_{\mathrm{out}}(x_1,x_2)
\end{equation}
The behavior of this state far from the support of $V_c$ is captured in the outgoing wavefunction $\psi_{\mathrm{out}}$. In general (and in particular, if the ingoing momentum is larger than the binding energy), this outgoing behavior now has two channels. In one channel we have elastic scattering, in which there is a single asymptotic state with the two particles still bound. The other channel captures the case of inelastic scattering, and there are two momenta. It is convenient to write these structures in integral form, from which the asymptotic behavior can be extracted via stationary-phase approximation:
\begin{equation}
	\begin{split}
	\psi_{\mathrm{out}}(x_1,x_2)
		= &{}
	\int\! dk\, e^{ikx} f_1(k) \theta(kx) \Psi(x_1-x_2)\\
		&+
	\int\! dk_1\,dk_2\, e^{i(k_1 x_1 + k_2 x_2)} f_2(k_1,k_2)
	\theta(k_1 x_1) \theta(k_2 x_2)
	\text.
	\end{split}
\end{equation}
As before we have used $x$ to denote the center-of-mass coordinate. Integrals are taken only over those $3$-momenta (or $d$-momenta in $d$ spatial dimensions) which respect conservation of energy. As only the asymptotic behavior of $\psi_{\mathrm{out}}$ matters, the Heaviside function $\theta(\cdot)$ can be replaced by any function of the same asymptotics.

The ingoing behavior uniquely determines the outgoing behavior for the same reason as in Theorem~\ref{thm:uniqueness}. Two states $\tilde\psi_1$ and $\tilde\psi_2$, with identical ingoing behavior, have a product $\phi$ which is an eigenstate with no ingoing current. The outgoing current of $\phi$ must therefore also vanish, implying that the outgoing waves of $\psi_1$ and $\psi_2$ are the same.

Now let $\psi$ be an exact scattering state, and $\tilde\psi$ an approximate one with the same ingoing wave. Each of these states is defined on $\mathbb R^3$. For any wavefunction $\xi$, Lemma~\ref{lem:holder} gives the inequality
\begin{equation}
	\|(E-\hat H)\tilde\psi\|_1
	\ge \|\xi\|_\infty^{-1} \int_{\mathbb R^3 \times \mathbb R^3} \xi^* (E - \hat H)\tilde\psi
	\text.
\end{equation}
Connecting the $L_1$ norm of the Schr\"odinger violation to the asymptotic error of $\tilde\psi-\psi$ is done by integrating by parts, and restricting $\xi$ to be an eigenstate, so that only boundary terms remain. In the case of two particles, the integral evaluates to
\begin{equation}
	\begin{split}
	2 M \|\xi\|_\infty
	\|(E - \hat H)\tilde\psi\|_1
		\ge {}&
	\lim_{R\rightarrow\infty}\int_{\mathbb R^3} dx_2
	\oint_{S^2_1}
	(\xi^* \nabla_1 \phi - \nabla_1 \xi^* \phi)\\
		&+
	\int_{\mathbb R^3} dx_1
	\lim_{R\rightarrow\infty}\oint_{S^2_2}
	(\xi^* \nabla_2 \phi - \nabla_2 \xi^* \phi)
	\text.
	\end{split}
\end{equation}
Above, $S^2_1$ and $S^2_2$ respectively denote the celestial spheres of the first and second particles. The corresponding integrals are evaluated at a radius $R$.
This calculation provides a starting point for proofs of stability of cross sections and phase shifts, just as in the case of a single particle scattering from a central potential.

We now turn to the general case. Let $V(x_1,\ldots,x_N)$ be a local $N$-body potential, so that the Hamiltonian of $A = A_1 + A_2$ particles is
\begin{equation}\label{eq:H-nbody}
	\hat H =
	\sum_{a=1}^A \frac{1}{2 M_a} \hat \nabla_a^2
	+
	\sum_a V(x_{a_1},\ldots,x_{a_N})
	\text.
\end{equation}
The second sum is taken over all sets of $N \le A$ particles, with indices $a_1,\ldots,a_N$. The generalization of Lemma~\ref{lem:stable} is as follows.
\begin{lemma}\label{lem:stable-bound}
	Let $\hat H$ be a corresponding Hamiltonian of the form of \eqref{H-nbody}, $k$ be a momentum, and $\xi$ be any eigenstate of $\hat H$ with eigenvalue $E$ for which $\|\xi\|_\infty$ is finite. Then given an exact scattering state $\psi$ and an approximate scattering state $\tilde\psi$, the $L_1$ norm of the Schr\"odinger violation of $\tilde\psi$ bounds the difference in their asymptotic behaviors according to:
	\begin{equation}
		\sum_a
		\int_{\mathbb R^{3(A-1)}} \lim_{R\rightarrow\infty}\oint_R (\xi^*\nabla_a\phi - \nabla_a \xi^*\phi)
		\le
		2 M \|\xi\|_\infty \|(E-\hat H)\tilde \psi\|_1
		\text.
	\end{equation}
\end{lemma}

Similar to the results of Lemma~\ref{lem:bounded}, scattering eigenstates in the presence of ingoing (or outgoing) bound states have bounded $L_\infty$ norms, given weak assumptions on the behavior of the potential. In the same manner as Lemma~\ref{lem:stable}, Lemma~\ref{lem:stable-bound} can be used to show that scattering properties are stable with respect to the $L_1$ norm of the approximate scattering state.

\subsection{Approximate scattering of approximate bound states}\label{sec:approx-approx}
In this section we briefly consider the case where the wavefunction of the desired ingoing bound state is not exactly known. We will sketch the central ideas by considering the case of elastic scattering of a two-body bound state from a central potential, as described by the Hamiltonian \eqref{H2-inelastic}. A more thorough treatment, along with the generalization to full inelastic scattering, is left for future work.

For brevity we will write this Hamiltonian as $\hat H = \hat H_0 + \hat K_{\mathrm{cm}} + \hat V_c$, with $\hat V_c$ containing the two terms coming from the central potential. The operator $\hat K_{\mathrm{cm}}$ contains the kinetic energy coming from center-of-mass motion. Therefore $\hat H_0$ has the property that the ingoing bound-state plane wave is an exact eigenstate. The energy of this eigenstate is the binding energy $E_b < 0$.

We denote the true bound-state wavefunction $\Psi$ as before, and the known approximation to the bound-state wavefunction is $\tilde\Psi$. As discussed above, the $L_1$ norm of the approximate scattering state is infinite, with divergent contributions both from the ingoing plane wave and the outgoing spherical wave. Concretely, decomposing the approximate scattering state in the usual way,
\begin{equation}
	\tilde \psi = \tilde \psi_{\mathrm{in}} + \tilde \psi_{\mathrm{bulk}} + \tilde\psi_{\mathrm{out}}
	\text,
\end{equation}
both $\tilde\psi_{\mathrm{in}}$ and $\tilde\psi_{\mathrm{out}}$ have the undesirable property that their image under $E-\hat H$ has infinite support.

The ingoing and outgoing waves are determined by $\tilde\Psi$ and the S-matrix amplitudes $\tilde f$:
\begin{equation}
	\tilde\psi_{\mathrm{in}}(x_1,x_2) = e^{ikx} \tilde\Psi(x_1-x_2)
	\text{ and }
	\tilde\psi_{\mathrm{out}}(x_1,x_2) = \frac{e^{ikr}}{r} \tilde f\left(\frac{x}{r}\right) \tilde\Psi(x_1-x_2) + O(r^{-2})\text.
\end{equation}
As in the one-body case, the precise definition of $\tilde \psi_{\mathrm{out}}$ is ambiguous, because terms can be shifted into $\tilde \psi_{\mathrm{bulk}}$ without changing the overall wavefunction.

 The $L_1$ norm, ordinarily thought of as a functional of $\tilde\psi$ itself, can also be considered a functional of $\tilde\psi_{\mathrm{bulk}}$, $\tilde f$, and the approximate bound-state wavefunction $\tilde\Psi$, as those ingredients are sufficient to uniquely construct a scattering state $\tilde\psi$ by applying the above definitions. Consider evaluating this functional at $\tilde\Psi = \Psi$. That is, take the $L_1$ norm of the Schr\"odinger violation of
\begin{equation}
	\tilde\psi_\Psi(x_1,x_2) =
	e^{ikx} \Psi(x_1-x_2)
	+
	\tilde\psi_{\mathrm{bulk}}(x_1,x_2)
	+
	\frac{e^{ikr}}{r} \tilde f\left(\frac{x}{r}\right) \Psi(x_1-x_2)\text.
\end{equation}
We will denote this $L_1$ norm $L[\tilde\psi_{\mathrm{bulk}},\tilde f] = \|(E-\hat H)\tilde \psi_\Psi\|_1$. This object is finite. In fact this is the $L_1$ loss that we would have evaluated if the true bound-state wavefunction was exactly known. Minimizing $L$ reveals the exact scattering state, and the stability theorems developed previously apply, so that $L[\cdot,\tilde f]$ being small implies that $\tilde f$ is close to the true $S$-matrix elements.

This construction is not immediately useful, because without knowledge of $\Psi$, there is no way to evaluate $L$. We now define a different modification of the $L_1$ norm, also finite, which can be evaluated by using $\tilde\Psi$. Loosely speaking, this is done by removing from the Schr\"odinger violation terms which, like $(E_b - \hat H_0)\tilde\Psi$, would vanish if $\tilde\Psi$ were equal to the exact bound state. The resulting functional, denoted $\tilde L$, is defined by
\begin{equation}
	\tilde L[\tilde\psi_{\mathrm{bulk}},\tilde f]
	=
	\int \left|
	(E - \hat H) \tilde \psi
	- (E_b - \hat H_0) \tilde \psi_{\mathrm{in}}
	- (E_b - \hat H_0) \tilde \psi_{\mathrm{out}}
	\right|
	\text.
\end{equation}
This is a very poor approximation to the original $L_1$ loss: the difference is infinite. However we will now show that $\tilde L$ can be a very good approximation to the ``ideal'' functional $L$, becoming equal when $\tilde\Psi = \Psi$.

Both $L$ and $\tilde L$ are defined as the integral of the absolute value of a sum of three terms. We will name these terms $\Lambda$ and $\tilde \Lambda$ with appropriate subscripts; thus $L = \int |\Lambda_{\mathrm{in}} + \Lambda_{\mathrm{bulk}} + \Lambda_{\mathrm{out}}|$ and $\tilde L$ is expanded similarly. The bulk wavefunction $\tilde\psi_{\mathrm{bulk}}$ enters identically in both expressions: $\Lambda_{\mathrm{bulk}} = \tilde \Lambda_{\mathrm{bulk}}$. The contributions of the ingoing plane waves are
\begin{equation}
	\Lambda_{\mathrm{in}}
	= -\hat V_c e^{ikx} \Psi(x_1 - x_2)
	\quad\text{and}\quad
	\tilde\Lambda_{\mathrm{in}} = (- \hat V_c) e^{ikx} \tilde\Psi(x_1 - x_2)
	\text,
\end{equation}
and therefore are related by $\tilde \Lambda_{\mathrm{in}} - \Lambda_{\mathrm{in}} = \hat V_c e^{ikx} (\Psi - \tilde \Psi)$. Similarly, the contributions from the outgoing waves are
\begin{equation}
	\begin{split}
	\Lambda_{\mathrm{out}}
		&= \left(
	E - E_b - \hat K_{\mathrm{cm}} - \hat V_c
		\right)
		\left(\frac{e^{i|k\|x|}}{|x|} \tilde f \left(\frac{x}{|x|}\right) + O(|x|^{-2})\right)
		\Psi(x_1-x_2)
		\text{ and}\\
	\tilde\Lambda_{\mathrm{out}}
		&= \left(
		E - E_b - \hat K_{\mathrm{cm}} - \hat V_c
		\right)
		\left(\frac{e^{i|k\|x|}}{|x|} \tilde f \left(\frac{x}{|x|}\right) + O(|x|^{-2})\right)
		\tilde\Psi(x_1-x_2)
	\text.
	\end{split}
\end{equation}
The center-of-mass Schr\"odinger violation operator, $E - E_b - \hat K_{\mathrm{cm}}$, does not annihilate the outgoing wavefunction, but instead yields terms which are higher-order in $|x|^{-1}$. The order of these terms depends on the choice of what terms to include $\tilde \psi_{\mathrm{out}}$. 
In general, for some integer $\alpha$, the difference between the two $L_1$ norms may be bounded as
\begin{equation}
	|\tilde L - L|
	\le
	\int \left|
	\hat V_c (\tilde \psi_{\mathrm{in}} - \psi_{\mathrm{in}})
	\right|
	+
	\int \left|
	\hat V_c (\tilde \psi_{\mathrm{out}} - \psi_{\mathrm{out}})
	\right|
	+
	\int \left|
	(\tilde \Psi(x_1-x_2) -
	\Psi(x_1-x_2))
	O(|x|^{-\alpha})
	\right|
	\text.
\end{equation}
This is finite as long as sufficiently high-order terms were included in the outgoing wavefunction, so as to make $\alpha \ge 4$. The right-hand side above can in turn be connected to measures of the error in the bound-state wavefunction, such as $\|\tilde\Psi - \Psi\|_1$. 

This bound allows the previous stability theorems to be applied even to scattering state for which the asymptotic wavefunctions are known only approximately. The stability theorems bound the asymptotic error in an approximate wavefunction in terms of the functional $L$, and the above results show that $L$ can be approximated by, and its value bounded in terms of, $\tilde L$.

\section{Discussion}\label{sec:discussion}

We have presented a minimum principle for scattering states; namely, that the scattering state is yielded by minimizing the $L_1$ norm of the Schr\"odinger violation $(E - \hat H) \tilde \psi$ over a class of wavefunctions $\tilde\psi$ constrained to obey certain boundary conditions. Furthermore we have shown that the (maximum possible) error in the asymptotic behavior of a scattering state $\tilde\psi$ is bounded by $C \int |(E-\hat H)\tilde\psi|$. The coefficient $C$ is potential-specific, and we are able to provide estimates only for a limited class of potentials. General estimates of this coefficient would also imply general bounds on scattering cross sections. This variational principle extends to momentum-dependent potentials (\secref{momentum}), long-range Coulomb-like potentials (\secref{coulomb}), and many-body scattering (\secref{many}).

There are notions of approximate scattering states that do not correspond to having small $L_1$ norm, and the bounds provided here are not necessarily close to tight for such approximations. As a simple example, consider in one dimension, the difference between the scattering states for potentials $V$ and $\tilde V$, defined as follows
\begin{equation}
	V(x) = \delta(x) \text{ and } \tilde V(x) = \delta(x-1)\text.
\end{equation}
Holding the incident momentum fixed, label the corresponding scattering states $\psi(x)$ and $\tilde\psi(x)$. Now $\tilde\psi$ may be considered an approximate scattering state for the potential $V(x)$, but the $L_1$ norm of the violation is not small. Nevertheless $\tilde\psi$ is in one sense an excellent approximation: the transmission coefficient read off from this approximate state is exactly correct!

An elaboration of this idea is given by Luty~\cite{Luty:RenormalizationNotes}. For asymptotically small incident momenta, the S-matrix provides no information about the structure of a potential other than its integral, $\int V(x)$. Thus the potential can be rearranged arbitrarily, as long as the integral is held fixed, without substantially modifying low-energy scattering. Now define, for any potential $V(x)$, an approximate scattering state $\tilde\psi$, as the \emph{exact} scattering state for the corresponding Dirac-$\delta$ potential:
\begin{equation}
	\tilde V(x) = \left(\int d^3 x\,V(x)\right)\delta(x)
	\text.
\end{equation}
At low incident momenta the state $\tilde\psi$ provides an excellent approximation to the S-matrix, but this fact cannot be seen from the $L_1$ norm of the Schr\"odinger violation, which may be large.

We have not discussed relativistic scattering states. In a gapped theory, these can be understood in essentially the same way that we treated many-body scattering. Each slice of Fock space which is expected to contribute can be parameterized in the same way, and the variational state optimized according the $L_1$ norm. We leave the proper treatment of this topic to future work, anticipating that some modification will be needed to accomodate the altered dispersion relation. We believe a sensible formulation of stable minimum principles for ungapped scattering states---where no finite number of sectors of Fock space is likely to be a good approximation---to be an open problem.

Decay of unstable bound states can be seen as a special sort of inelastic scattering. Decay is a $1\rightarrow N$ process instead of a $2\rightarrow N$ process, and more importantly has dramatically different boundary conditions, in which the outgoing waves must have complex momentum. The results derived in this work do not appear to have any straightforward analogues to the case of decay states, and any generalization is left to future work.

\section*{Acknowledgments}

For many insightful conversations on scattering and variational methods, we are grateful to Tanmoy Bhattacharya, Tom Cohen, Duff Neill, and Ingo Tews.

S.L.~is supported by a Richard P.~Feynman fellowship from the LANL LDRD program, and Y.Y.~is supported by a Darleane C.~Hoffman fellowship from the LANL LDRD program. Los Alamos National Laboratory is operated by Triad National Security, LLC, for the National Nuclear Security Administration of U.S. Department of Energy (Contract No.~89233218CNA000001).

\printbibliography

\appendix
\section{Examples and counterexamples}\label{app:counterexamples}
In this appendix we collect a few pathological approximate scattering states, which have served as useful landmarks in constructing the results in the body of this paper. We include them here both to justify choices made in this work, and as a guide to those who wish to extend our results.

\paragraph{Justifying the $L_1$ norm} Our first counterexample is intended to explain why we are interested in the $L_1$ norm of the Schr\"odinger violation $(E-\hat H)\tilde\psi$, rather than the $L_2$ norm. We will construct a sequence of states $\Psi_n$, for which the $L_2$ norm of the Schr\"odinger violation $\int |(E-\hat H)\Psi_n|^2$ is arbitrarily close to $0$, but $\int \nabla \cdot j$ remains of order unity. For these examples we select $\hat H$ to be the free Hamiltonian.

First define the triangle function (although the precise form doesn't much matter):
\begin{equation}
	T(x) = \begin{cases}
		0 & x < 0\\
		x & x \in [0,1]\\
		1 & x > 1
		\text.
	\end{cases}
\end{equation}
Now take $\Psi_n$ to be a state which changes amplitude, from $A_<$ as $x \rightarrow -\infty$ to $A_>$ as $x \rightarrow \infty$, with the transition defined by $T(x/n)$.
\begin{equation}
	\Psi_n(x) = \theta(-x) A_< e^{ikx} + T(x/n) A_> e^{ikx}
	\text.
\end{equation}
The integral of the divergence of the current is $A_> - A_<$, which is independent of $n$. The $L_2$ norm of the Schr\"odinger violation is proportional to $n^{-1}$, and can be made arbitrarily small while holding $A_> - A_<$ fixed.

The motivation for this counterexample is that $\int \nabla \cdot j$ measures the total amount of probability being created. For the example here, $\Psi$ looks much like a scattering state, but with an apparently nonunitary S-matrix.

The existence of this counterexample is connected to the fact that we are measuring the standard deviation of $(E-\hat H)$ in an ungapped system, and also to the fact that the $L_2$ norm is quadratic in $\psi$. Some variational principles (see the brief review in \secref{prior}) are based instead on the expectation value of $(E - \hat H)$. A similar counterexample also shows that such variational principles are not stable: $\langle\Psi|(E-\hat H)|\Psi\rangle$ can be made arbitrarily small while the S-matrix remains arbitrarily wrong. Here we need only ensure that $(E-\hat H)\psi$ be orthogonal, or nearly orthogonal, to $\psi$. One possible sequence of states is
\begin{equation}
	\Psi_n(x) = \theta(-x) A_< \sin k x + T(\alpha x/n) A_> \sin k x
	\text.
\end{equation}
At sufficiently large $n$, the parameter $\alpha$ (of order unity) can be tuned to make the expectation value of $E - \hat H$ exactly vanish. Note that there's no need to take a wavefunction which is nontrivial over an arbitrarily large region. This result is unsurprising: the spectrum of $\hat H$ includes infinitely many states with energies both above and below $E$, and so there should be no shortage of superpositions satisfying $\langle \hat H \rangle = E$.

\paragraph{Large probability nonconservation from small Schr\"odinger violation} This second counterexample concerns probability nonconservation in approximate eigenstates. The fact that probability is conserved in an eigenstate is central to the scattering uniqueness theorem (Theorem~\ref{thm:uniqueness}). The reader may wonder why strengthening the theorem to account for a small violation of Schr\"odinger's equation required such a dramatic change in the proof, including the abandonment of any discussion of $\nabla\cdot j$. It is precisely because of the following construction, of a sequence of states for which the $L_1$ norm of the Schr\"odinger violation is arbitrarily small, but the total probability nonconservation $\int \nabla \cdot j$ is large.

Define the two-parameter family of wavefunctions
\begin{equation}
	\Psi_{\epsilon,A}(x) = i A \theta(-x) \sin k x
	+ T(x)\left[i A \sin k x + \epsilon \cos k x\right]
	\text.
\end{equation}
The current at $x < 0$ vanishes for the usual reason. The Schr\"odinger violation, coming entirely from the interval $[0,1]$, is proportional to $\epsilon$ and independent of $A$. The current at $x > 1$ is
\begin{equation}
	j(x>1) = \frac{\Im \left(-i A \sin k x + \epsilon \cos k x\right)k \left(i A \cos k x - \epsilon \sin k x\right)}{M}
	= \frac{\epsilon k A}{M}
	\text.
\end{equation}
By taking $A = \epsilon^{-1}$, we obtain large probability nonconservation while the Schr\"odinger violation is arbitrarily small.

\paragraph{The relevance of pointwise bounds} Our third (counter)example is intended to provide some intuition for the relevance of the $L_{\infty}$ norm of the exact scattering state. When this norm is large, the constant $C$ that appears in Lemma~\ref{lem:stable} (and analogous subsequent lemmas) is correspondingly large, resulting in weak bounds on the error of approximate scattering states. At least in one example, this is not a technical limitation of the method of proof used, but rather a real fact about approximate scattering states in certain potentials.

\begin{figure}
	\centering
	\includegraphics[width=0.7\textwidth]{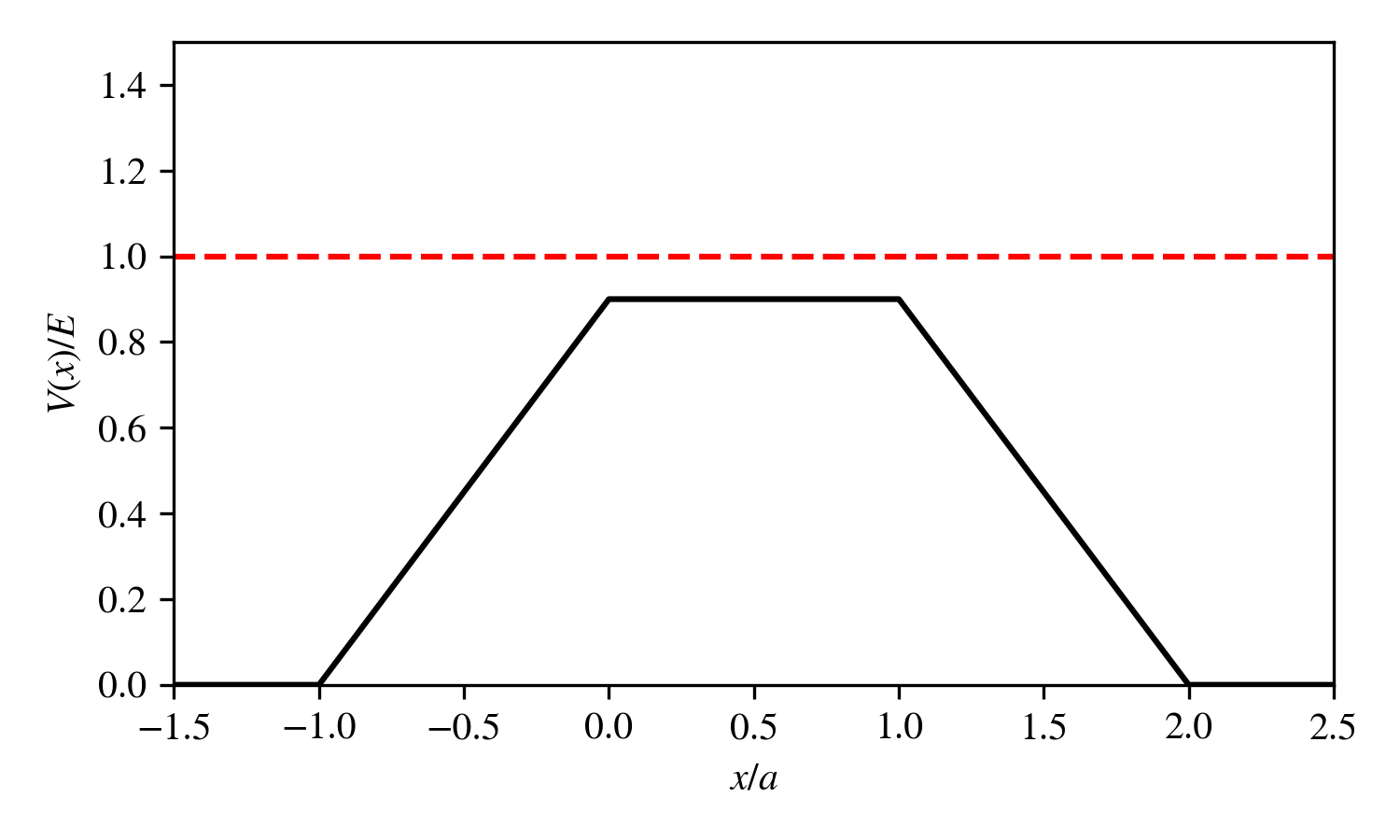}
	\caption{A schematic of the potential defined by (\ref{eq:V-slow}). The ingoing particle ``slows down'' dramatically in the interval $x \in (0,a)$, where $E-V(x)$ is small, and the scattering wavefunction $\psi(x)$ is correspondingly large there.\label{fig:V-slow}}
\end{figure}

For this example we consider the scattering of a particle with energy $E$ off of a potential whose maximum is very close to, and just below, $E$. The potential, also depicted in Figure~\ref{fig:V-slow}, is:
\begin{equation}\label{eq:V-slow}
	V(x;a,\epsilon) = \begin{cases}
		(1-\epsilon)\left(\frac x a+1\right) & x \in [-a,0]\\
		1-\epsilon & x \in (0,a)\\
		(1-\epsilon)\left(2-\frac x a\right)& x \in [a,2a]\\
		0 & \text{elsewhere.}
	\end{cases}
\end{equation}
This potential is parameterized by two quantities, $\epsilon$ and $a$, with dimensions of energy and length respectively. For fixed $\epsilon > 0$, we can always take $a$ sufficiently large so that scattering may be analyzed semiclassically. Focusing on the interval $[0,a]$, the scattering state with ingoing wave $\psi(x<-a) = e^{i\sqrt{2ME}x}$ obeys
\begin{equation}
	\psi(x) = A e^{i\theta} e^{i\sqrt{2M\epsilon}x} \text{ where } x\in[0,a]\text,
\end{equation}
where $A$ and $\theta$ are both real. The overall phase $\theta$ does not concern us. The amplitude $A$ may be determined by probability conservation. The velocity of the particle (of mass $M$) outside of the potential is $\sqrt{2E/M}$; where the potential is maximized, the velocity is $\sqrt{2\epsilon/M}$. The ratio of probability densities is therefore $\sqrt{\epsilon/E}$, and so we find $A = (E/\epsilon)^{1/4}$.

Now we consider approximate scattering states, and in particular states with the same ingoing wave as above $(\phi(x<-a) = e^{i\sqrt{2M E}x})$, but no outgoing wave $(\phi(x>2a) = 0)$. We will place all of the Schr\"odinger violation on the interval $[0,a]$, so that our approximate scattering state has the form
\begin{equation}
	\tilde\psi(x) = \begin{cases}
		\psi(x) & x < 0\\
		f(x) & x \in [0,a]\\
		0 & x > a
		\text.
	\end{cases}
\end{equation}
The choice of transition function $f(x)$ is now fairly arbitrary; for ease of computation we choose
\begin{equation}
	f(x)
	= \frac{a-x}{a}\left(\frac E \epsilon\right)^{-1/4} e^{i \sqrt{2 M \epsilon} x}
	\text.
\end{equation}
The $L_1$ norm of the Schr\"odinger violation has contributions from each of the two kinks, at $x=0$ and $x=a$, where the transition function $f$ turns on and off. These contributions are proportional to $a^{-1}$, and can therefore be made arbitrarily small at fixed $\epsilon$. In this limit the dominant contribution to the Schr\"odinger violation comes from the integral over $(0,a)$:
\begin{equation}
	\int_0^a |(E-\hat H)f(x)| \,dx
	=
	\left(\frac{E}{\epsilon}\right)^{1/4}\displaystyle\int_0^a \left|\left(\epsilon + \frac{\partial^2}{2M}\right)\frac{a-x}{a} e^{i\sqrt{2M\epsilon}x}\right|\, dx
	= \sqrt{\frac 2 M} E^{1/4} \epsilon^{1/4}
	\text.
\end{equation}
This contribution is $a$-independent and cannot be reduced by spreading the potential out. Critically, as $\epsilon$ grows small, we find that the $L_1$ norm of the violation also grows small, with precisely the power-law decay predicted from the $L_\infty$ norm of the exact scattering state $\psi$ above.

\paragraph{Scattering states without pointwise bounds} Finally, we address the conditions under which the scattering states of a Hamiltonian cannot be pointwise bounded; that is, the conditions under which Lemma~\ref{lem:bounded} does not hold. In one dimension, define a potential on the positive real line $\mathbb R_{++}$:
\begin{equation}
	V(r) =\alpha r^{-2} 
	\text.
\end{equation}
For a particle of mass $m$, this potential has notably pathological behavior for $\alpha < - \frac 1 {8m}$, including a spectrum unbounded from below~\cite{Essin:2006sic}.

We take as the ansatz for an eigenstate $\psi = r^\beta (1 + O(r))$. Note that we are not concerned with requiring that $\psi$ be square-integrable; scattering states are, by construction, never square-integrable. At leading order in $r$, Schr\"odinger's equation is satisfied as long as
\begin{equation}
	2 m \alpha = \beta(\beta - 1)
	\text.
\end{equation}
This is true for any energy $E$---only higher-order terms in $r$ depend on the energy of the state. However, for all $\alpha > 0$ there exist eigenstates which have divergences at $r=0$, and therefore fail to be pointwise bounded.

Similar calculations show that this form of potential is not problematic for any exponent other than $-2$. In particular, Yukawa and Coulomb potentials have an exponent of $-1$ and do not exhibit this sort of behavior.

\end{document}